\documentclass{article}

\usepackage[round]{natbib}
\usepackage{amsmath,amsthm}
\usepackage{amssymb}

\usepackage[T1]{fontenc}

\newtheorem{theorem}{Theorem}[section]
\newtheorem{corollary}[theorem]{Corollary}
\newtheorem{lemma}[theorem]{Lemma}

\theoremstyle{definition}
\newtheorem{definition}[theorem]{Definition}

\theoremstyle{remark}
\newtheorem*{remark*}{Remark}

\numberwithin{equation}{section}

\newcommand{\treeo}{\circle*{2}}

\newcommand{\treei}[2]{\begin{picture}(0,0)
\put(0,0){\circle*{2}}
\put(0,0){\line(-1,-1){10}}
\put(0,0){\line(1,-1){10}}
\put(-10,-10){#1}
\put(10,-10){#2}
\end{picture}}

\newcommand{\treeii}[2]{\begin{picture}(0,0)
\put(0,0){\circle*{2}}
\put(0,0){\line(-1,-2){5}}
\put(0,0){\line(1,-2){5}}
\put(-5,-10){#1}
\put(5,-10){#2}
\end{picture}}

\newcommand{\treeiii}[2]{\begin{picture}(0,0)
\put(0,0){\circle*{2}}
\put(0,0){\line(-1,-4){2.5}}
\put(0,0){\line(1,-4){2.5}}
\put(-2.5,-10){#1}
\put(2.5,-10){#2}
\end{picture}}

\author{Matthew P. Szudzik}

\title{The Rosenberg-Strong Pairing Function}

\date{2019-01-28}

\begin{document}

\maketitle

\begin{abstract}
This article surveys the known results (and not very well-known results) associated with Cantor's pairing function and the Rosenberg-Strong pairing function, including their inverses, their generalizations to higher dimensions, and a discussion of a few of the advantages of the Rosenberg-Strong pairing function over Cantor's pairing function in practical applications.  In particular, an application to the problem of enumerating full binary trees is discussed.
\end{abstract}

\section{Cantor's pairing function}

Given any set $B$, a \emph{pairing function}\footnote{
There is no general agreement on the definition of a pairing function in the published literature.  For a given set $B$, some publications~\citep{Barendregt1974,Rosenberg1975b,Simpson1999} use a more general definition, and allow any \emph{one-to-one} function (\textit{i.e.}, any injection) from $B^2$ to $B$ to be regarded as a pairing function.  Other publications~\citep{Regan1992,Rosenberg2003} are more restrictive, and require each pairing function to be a \emph{one-to-one correspondence} (\textit{i.e.}, a bijection).  We use the more restrictive definition in this paper.
%
%
%
} for $B$ is a one-to-one correspondence from the set of ordered pairs $B^2$ to the set $B$.  The only finite sets $B$ with pairing functions are the sets with fewer than two elements.  But if $B$ is infinite, then a pairing function for $B$ necessarily exists.\footnote{
The proof~\citep{Zermelo1904} for this claim relies on the axiom of choice.  In fact, in Zermelo set theory, asserting that all infinite sets have pairing functions is equivalent to asserting the axiom of choice~\citep{Tarski1924}.
}  For example, \emph{Cantor's pairing function}~\citep{Cantor1878} for the positive integers is the function
\begin{equation*}
p(x,y)=\frac{1}{2}(x^2+2xy+y^2-x-3y+2)
\end{equation*}
that maps each pair $(x,y)$ of positive integers to a single positive integer $p(x,y)$.  Cantor's pairing function serves as an important example in elementary set theory~\citep{Enderton1977}.  It is also used as a fundamental tool in recursion theory and other related areas of mathematics~\citep{Rogers1967,Matiyasevich1993}.

A few different variants of Cantor's pairing function appear in the literature.  First, given any pairing function $f(x,y)$ for the positive integers, the function $f(x+1,y+1)-1$ is a pairing function for the non-negative integers.  Therefore, we refer to the function
\begin{equation*}
c(x,y)=p(x+1,y+1)-1=\frac{1}{2}(x^2+2xy+y^2+3x+y)
\end{equation*}
as \emph{Cantor's pairing function} for the non-negative integers.  And given any pairing function $f(x,y)$ for a set $B$, the function obtained by exchanging $x$ and $y$ in the definition of $f$ is itself a pairing function for $B$.  Hence,
\begin{equation*}
\bar{c}(x,y)=\frac{1}{2}(y^2+2yx+x^2+3y+x)
\end{equation*}
is another variant of Cantor's pairing function for the non-negative integers.  It has been shown by Fueter and P\'{o}lya~\citeyearpar{Fueter1923} that there are only two quadratic polynomials that are pairing functions for the non-negative integers, namely the polynomials $c(x,y)$ and $\bar{c}(x,y)$.  But it is a longstanding open problem whether there exist any other polynomials, of higher degree, that are pairing functions for the non-negative integers.  Partial results toward a resolution of this problem have been obtained by Lew and Rosenberg~\citeyearpar{Lew1978a,Lew1978b}.

An \emph{enumeration} of a countably infinite set $C$ is a one-to-one correspondence from the set $\mathbb{N}$ of non-negative integers to the set $C$.  We think of an enumeration $g\colon\mathbb{N}\to C$ as \emph{ordering} the members of $C$ in the sequence
\begin{equation*}
g(0),\;\;g(1),\;\;g(2),\;\;g(3),\;\;\ldots.
\end{equation*}
Given any pairing function $f\colon\mathbb{N}^2\to\mathbb{N}$, its inverse $f^{-1}\colon\mathbb{N}\to\mathbb{N}^2$ is an enumeration of the set $\mathbb{N}^2$.  For example, the inverse of Cantor's pairing function $c(x,y)$ orders the points in $\mathbb{N}^2$ according to the sequence
\begin{equation}\label{c-seq}
(0,0),\;\;(0,1),\;\;(1,0),\;\;(0,2),\;\;(1,1),\;\;(2,0),\;\;\ldots.
\end{equation}
This sequence is illustrated in Figure~\ref{c-illust}.
\begin{figure}
\begin{center}
\begin{picture}(167,171)(-29,-32)
%
%
\multiput(0,0)(0,36){4}{\multiput(0,0)(36,0){4}{\circle*{4}}}
%
\thinlines
\put(-18,-18){\line(1,0){144}}
\put(-18,-18){\line(0,1){144}}
%
\put(132,-18){\makebox(0,0)[l]{$x$}}
\put(-18,132){\makebox(0,0)[b]{$y$}}
%
\thinlines
\multiput(0,-18)(36,0){4}{\line(0,1){3}}
\multiput(-18,0)(0,36){4}{\line(1,0){3}}
%
\put(0,-25){\makebox(0,0)[t]{0}}
\put(36,-25){\makebox(0,0)[t]{1}}
\put(72,-25){\makebox(0,0)[t]{2}}
\put(108,-25){\makebox(0,0)[t]{3}}
%
\put(-24,0){\makebox(0,0)[r]{0}}
\put(-24,36){\makebox(0,0)[r]{1}}
\put(-24,72){\makebox(0,0)[r]{2}}
\put(-24,108){\makebox(0,0)[r]{3}}
%
\thicklines
\multiput(5,31)(36,-36){1}{\vector(1,-1){26}}
\multiput(5,67)(36,-36){2}{\vector(1,-1){26}}
\multiput(5,103)(36,-36){3}{\vector(1,-1){26}}
%
\put(3,4){\makebox(0,0)[bl]{0}}
\put(3,40){\makebox(0,0)[bl]{1}}
\put(39,4){\makebox(0,0)[bl]{2}}
\put(3,76){\makebox(0,0)[bl]{3}}
\put(39,40){\makebox(0,0)[bl]{4}}
\put(75,4){\makebox(0,0)[bl]{5}}
\put(3,112){\makebox(0,0)[bl]{6}}
\put(39,76){\makebox(0,0)[bl]{7}}
\put(75,40){\makebox(0,0)[bl]{8}}
\put(111,4){\makebox(0,0)[bl]{9}}
\end{picture}
\caption{Cantor's pairing function $c(x,y)$.}
\label{c-illust}
\end{center}
\end{figure}
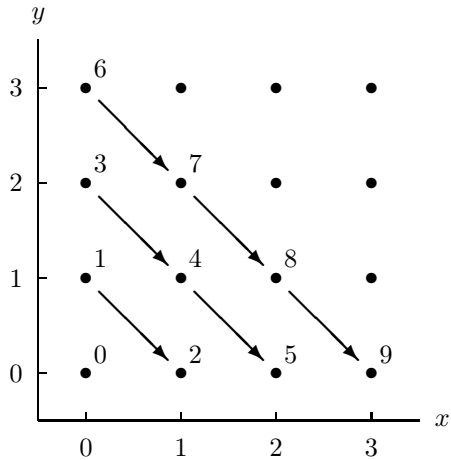
Cantor's pairing function is closely related to Cauchy's product formula~\citep{Cauchy1821}, which defines the \emph{product} of two infinite series, $\sum_{i=0}^\infty a_i$ and $\sum_{i=0}^\infty b_i$, to be the infinite series
\begin{equation*}
\sum_{i=0}^\infty\sum_{j=0}^i a_jb_{i-j}=a_0b_0+(a_0b_1+a_1b_0)+(a_0b_2+a_1b_1+a_2b_0)+\cdots.
\end{equation*}
In particular, the pairs of subscripts in this sum occur in exactly the same order as the points in sequence~\eqref{c-seq}.

Given any point $(x,y)$ in $\mathbb{N}^2$, we say that the quantity $x+y$ is the point's \emph{shell number} for Cantor's pairing function.  In Figure~\ref{c-illust}, any two consecutive points that share the same shell number have been joined with an arrow.  The inverse of Cantor's pairing function $c(x,y)$ is given by the formula
\begin{equation}\label{c-inv}
c^{-1}(z)=\left(z-\frac{w(w+1)}{2},\frac{w(w+3)}{2}-z\right),
\end{equation}
where
\begin{equation*}
w=\left\lfloor\frac{-1+\sqrt{1+8z}}{2}\right\rfloor,
\end{equation*}
and where, for all real numbers $t$, $\lfloor t\rfloor$ denotes the floor of $t$.  In his derivation of this inverse formula, Davis~\citeyearpar{Davis1958} has shown that $w$ is the shell number of $c^{-1}(z)$.  In fact, we can deduce this directly from equation~\eqref{c-inv}, since the two components of $c^{-1}(z)$ sum to $w$.
%
%

In elementary set theory~\citep{Enderton1977}, an ordered triple of elements $(x_1,x_2,x_3)$ is defined as an abbreviation for the formula $\bigl((x_1,x_2),x_3\bigr)$.  Likewise, an ordered quadruple $(x_1,x_2,x_3,x_4)$ is defined as an abbreviation for $\bigl(\bigl((x_1,x_2),x_3\bigr),x_4\bigr)$, and so on.  A similar idea can be used to generalize pairing functions to higher dimensions.  Given any set $B$ and any positive integer $d$, we say that a one-to-one correspondence from $B^d$ to $B$ is a \emph{$d$-tupling function} for $B$.  For example, if $f$ is any pairing function for a set $B$, then
\begin{equation}\label{tripling}
g(x_1,x_2,x_3)=f\bigl(f(x_1,x_2),x_3\bigr)
\end{equation}
is a $3$-tupling function for $B$.  The inverse of this $3$-tupling function is
\begin{equation*}
g^{-1}(z)=\bigl(f^{-1}(u),v\bigr),
\end{equation*}
where $u$ and $v$ are defined so that $(u,v)=f^{-1}(z)$.

Cantor's pairing function $c(x_1,x_2)$ is a quadratic polynomial pairing function.  By equation~\eqref{tripling}, the $4$th degree polynomial $c\bigl(c(x_1,x_2),x_3\bigr)$ is a $3$-tupling function.  And similarly, the $8$th degree polynomial $c\bigl(c\bigl(c(x_1,x_2),x_3\bigr),x_4\bigr)$ is a $4$-tupling function.  By repeatedly applying Cantor's pairing function in this manner, one can obtain a $(2^{d-1})$-degree polynomial $d$-tupling function, given any positive integer $d$.  This method for generalizing Cantor's pairing function to higher dimensions is commonly used in recursion theory~\citep{Rogers1967}.

Another higher-dimensional generalization of Cantor's pairing function was identified by Skolem~\citeyearpar{Skolem1937}.  In particular, for each positive integer $d$, the $d$-degree polynomial
\begin{equation*}
s_d(x_1,x_2,\ldots,x_d)=\sum_{i=1}^d\binom{x_1+x_2+\cdots+x_i+i-1}{i}
\end{equation*}
is a $d$-tupling function for the non-negative integers.  Note that $s_2(x_1,x_2)$ is Cantor's pairing function for the non-negative integers.  The inverse function $s_d^{-1}(z)$ can be calculated using the $d$-canonical
representation\footnote{
An algorithm for calculating the $d$-canonical representation is described by Kruskal~\citeyearpar{Kruskal1963}.  Several different names for this representation appear in the published literature, including the \emph{combinatorial number system}~\citep{Knuth1997}, the \emph{$d$th Macaulay representation}~\citep{Green1989}, the \emph{$d$-binomial expansion}~\citep{Greene1978}, and the \emph{$d$-cascade representation}~\citep{Frankl1984}.  The earliest known description of the $d$-canonical representation is in a paper by Ernesto Pascal~\citeyearpar{Pascal1887}.
} of $z$.

Lew~\citeyearpar[pp. 265--266]{Lew1979} has shown that every polynomial $d$-tupling function for the non-negative integers has degree $d$ or higher.  That is, Skolem's $d$-tupling function has the smallest possible degree for a polynomial $d$-tupling function from $\mathbb{N}^d$ to $\mathbb{N}$.  But $s_d(x_1,x_2,\ldots,x_d)$ is not necessarily the only polynomial $d$-tupling function with degree $d$.
%
%

Let $A_d$ denote the set of all $d$-tupling functions for the non-negative integers that can be expressed as polynomials of degree $d$.  Given any $d$-tupling function for a set $B$, a function obtained by permuting the order of its arguments is also a $d$-tupling function for $B$.
%
%
If this permutation is not the identity, and if $B$ contains at least two elements,
%
%
then the $d$-tupling function obtained in this manner is necessarily distinct from the original function.
%
%
Therefore, the set $A_d$ contains at least $d\,!$ many functions---one for each permutation of the arguments of $s_d(x_1,x_2,\ldots,x_d)$.  If $d=1$ or $d=2$, then these are the only members of $A_d$, and $\lvert A_d\rvert=d\,!$ in this case.  But Chowla~\citeyearpar{Chowla1961} has shown that for all positive integers $d$,
\begin{multline*}
\chi_d(x_1,x_2,\ldots,x_d)=\\
\binom{x_1+\cdots+x_d+d}{d}-1-\sum_{i=1}^{d-1}\binom{x_{i+1}+\cdots+x_d+d-i-1}{d-i}
\end{multline*}
is also a $d$-degree polynomial $d$-tupling function for the non-negative integers.\footnote{
The function originally described by Chowla was a $d$-tupling function for the positive integers.  The variant given here is obtained by translating Chowla's function from the positive integers to the non-negative integers.
}  Chowla's function is identical to Skolem's function  for $d=1$ and $d=2$.  When $d$ is greater than $2$, the following theorem holds.

\begin{theorem}
Let $d$ be any integer greater than $2$.  Then, $\chi_d(x_1,x_2,\ldots,x_d)$ cannot be obtained by permuting the arguments of $s_d(x_1,x_2,\ldots,x_d)$.
\end{theorem}
\begin{proof}[Proof by contradiction]
Assume that $\chi_d(x_1,x_2,\ldots,x_d)$ can be obtained by permuting the arguments of $s_d(x_1,x_2,\ldots,x_d)$.  Then, since $s_d\colon\mathbb{N}^d\to\mathbb{N}$ is a one-to-one correspondence and
\begin{equation*}
s_d(0,1,0,0,\ldots,0)=\chi_d(0,1,0,0,\ldots,0),
\end{equation*}
the permutation must not move the argument $x_2$.  But this contradicts the fact that
\begin{equation*}
s_d(0,2,0,0,\ldots,0)\ne\chi_d(0,2,0,0,\ldots,0).
\end{equation*}
Therefore, the assumption is false, and $\chi_d(x_1,x_2,\ldots,x_d)$ cannot be obtained by permuting the arguments of $s_d(x_1,x_2,\ldots,x_d)$.
\end{proof}

An immediate consequence is that if $d$ is greater than $2$, then $\lvert A_d\rvert>d\,!$.  A more precise lower bound for $\lvert A_d\rvert$ has been provided by
%
%
Morales and Arredondo \citeyearpar{Morales1999}.  In particular, they prove that
\begin{equation}\label{A-card}
\lvert A_d\rvert\geq d\,!\,a(d),
\end{equation}
where $a(d)$ is defined recursively so that $a(1)=1$, and so that
\begin{equation*}
a(d)=\sum_{\substack{i\in\mathbb{N},\;i\ne1\\i\;\text{divides}\;d}}(i-1)!\,a(d/i)
\end{equation*}
for all integers $d>1$.  They also construct a set of polynomial $d$-tupling functions for each positive integer $d$.  Morales and Arredondo conjecture that this is the set of all polynomial $d$-tupling functions for the non-negative integers.  If their conjecture is correct, then the inequality~\eqref{A-card} is an equality.  A proof of the Morales-Arredondo conjecture would also imply that $c(x,y)$ and $\bar{c}(x,y)$ are the only polynomial pairing functions for the non-negative integers.
%
%

\section{Other pairing functions}

In addition to Cantor's pairing function, a few other pairing functions are often encountered in the literature.  Most notably, the \emph{Rosenberg-Strong pairing function}~\citep{Rosenberg1972,Rosenberg1974} for the non-negative integers is defined by the formula\footnote{
The function originally described by Rosenberg and Strong was a pairing function for the positive integers.  The variant defined here is obtained by translating the Rosenberg-Strong function from the positive integers to the non-negative integers, and by reversing the order of the arguments.
}
\begin{equation}\label{r2}
r_2(x,y)=\bigl(\max(x,y)\bigr)^2+\max(x,y)+x-y.
\end{equation}
In the context of the Rosenberg-Strong pairing function, the quantity $\max(x,y)$ is said to be the \emph{shell number} of the point $(x,y)$.  Figure~\ref{r-illust} contains an illustration of the Rosenberg-Strong pairing function.
\begin{figure}
\begin{center}
\begin{picture}(167,171)(-29,-32)
%
%
\multiput(0,0)(0,36){4}{\multiput(0,0)(36,0){4}{\circle*{4}}}
%
\thinlines
\put(-18,-18){\line(1,0){144}}
\put(-18,-18){\line(0,1){144}}
%
\put(132,-18){\makebox(0,0)[l]{$x$}}
\put(-18,132){\makebox(0,0)[b]{$y$}}
%
\thinlines
\multiput(0,-18)(36,0){4}{\line(0,1){3}}
\multiput(-18,0)(0,36){4}{\line(1,0){3}}
%
\put(0,-25){\makebox(0,0)[t]{0}}
\put(36,-25){\makebox(0,0)[t]{1}}
\put(72,-25){\makebox(0,0)[t]{2}}
\put(108,-25){\makebox(0,0)[t]{3}}
%
\put(-24,0){\makebox(0,0)[r]{0}}
\put(-24,36){\makebox(0,0)[r]{1}}
\put(-24,72){\makebox(0,0)[r]{2}}
\put(-24,108){\makebox(0,0)[r]{3}}
%
\thicklines
\multiput(6,36)(36,0){1}{\vector(1,0){24}}
\multiput(6,72)(36,0){2}{\vector(1,0){24}}
\multiput(6,108)(36,0){3}{\vector(1,0){24}}
%
\thicklines
\multiput(36,30)(0,36){1}{\vector(0,-1){24}}
\multiput(72,30)(0,36){2}{\vector(0,-1){24}}
\multiput(108,30)(0,36){3}{\vector(0,-1){24}}
%
\put(5,5){\makebox(0,0)[bl]{0}}
\put(5,41){\makebox(0,0)[bl]{1}}
\put(41,41){\makebox(0,0)[bl]{2}}
\put(41,5){\makebox(0,0)[bl]{3}}
\put(5,77){\makebox(0,0)[bl]{4}}
\put(41,77){\makebox(0,0)[bl]{5}}
\put(77,77){\makebox(0,0)[bl]{6}}
\put(77,41){\makebox(0,0)[bl]{7}}
\put(77,5){\makebox(0,0)[bl]{8}}
\put(5,113){\makebox(0,0)[bl]{9}}
\put(41,113){\makebox(0,0)[bl]{10}}
\put(77,113){\makebox(0,0)[bl]{11}}
\put(113,113){\makebox(0,0)[bl]{12}}
\put(113,77){\makebox(0,0)[bl]{13}}
\put(113,41){\makebox(0,0)[bl]{14}}
\put(113,5){\makebox(0,0)[bl]{15}}
\end{picture}
\caption{The Rosenberg-Strong pairing function $r_2(x,y)$.}
\label{r-illust}
\end{center}
\end{figure}
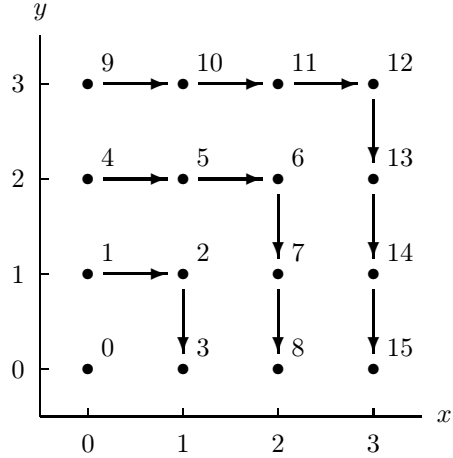
In this illustration, points that appear consecutively in the enumeration $r_2^{-1}\colon\mathbb{N}\to\mathbb{N}^2$ are joined by an arrow if and only if they share the same shell number.  The inverse of the Rosenberg-Strong pairing function $r_2(x,y)$ is given by the formula
\begin{equation*}
r_2^{-1}(z)=\begin{cases}
\bigl(z-m^2,m\bigr) &\text{if $z-m^2<m$}\\
\bigl(m,m^2+2m-z\bigr) &\text{otherwise}\rule{0pt}{12pt}
\end{cases},
\end{equation*}
where $m=\bigl\lfloor\sqrt{z}\,\bigr\rfloor$.  Note that $m$ is the shell number of the point $r_2^{-1}(z)$.

A shell numbering can serve as a useful tool when one wishes to describe the properties of a $d$-tupling function.  Formally, we make the following definition. 
\begin{definition}\label{shell-numbering}
Let $f\colon\mathbb{N}^d\to\mathbb{N}$ be a $d$-tupling function.  A function $\sigma\colon\mathbb{N}^d\to\mathbb{N}$ is said to be a \emph{shell numbering} for $f$ if and only if
\begin{equation*}
\sigma(\mathbf{x})<\sigma(\mathbf{y})\quad\text{implies}\quad f(\mathbf{x})<f(\mathbf{y})
\end{equation*}
for all $\mathbf{x}$ and $\mathbf{y}$ in $\mathbb{N}^d$.
\end{definition}

Given a shell numbering $\sigma$, the quantity $\sigma(\mathbf{x})$ is said to be the \emph{shell number} of $\mathbf{x}$.  The \emph{shell} that contains $\mathbf{x}$ is the set of all points in $\mathbb{N}^d$ that have the same shell number as $\mathbf{x}$.  Therefore, a shell numbering $\sigma$ partitions $\mathbb{N}^d$ into shells according to the equivalence relation $\sigma(\mathbf{x})=\sigma(\mathbf{y})$.  Notice that each $d$-tupling function for the non-negative integers has more than one shell numbering.  In particular, given any $d$-tupling function $f\colon\mathbb{N}^d\to\mathbb{N}$, the constant functions $\sigma(\mathbf{x})=k$, where $k$ is a non-negative integer, are all shell numberings for $f$.  But there is often one particular shell numbering that we prefer to use in the context of a given $d$-tupling function.  We call this the \emph{standard} shell numbering, or simply \emph{the} shell numbering, for the function.

For Cantor's pairing function $c(x,y)$, the standard shell numbering is
%
%
$\sigma(x,y)\linebreak[0]=x+y$, and each shell is a set of points on a diagonal line.  For this reason, $d$-tupling functions with the shell numbering $\delta(x_1,x_2,\ldots,x_d)=x_1+x_2+\cdots+x_d$ are said to have \emph{diagonal shells}.\footnote{
Of course, when $d>2$ the points in each shell lie on a plane or hyperplane, rather than a line.
}  The pairing functions $c(x,y)$ and $\bar{c}(x,y)$ both have diagonal shells.
%
%
The $d$-tupling functions $s_d(x_1,x_2,\ldots,x_d)$ and $\chi_d(x_1,x_2,\ldots,x_d)$ also have diagonal shells.\footnote{
Morales and Lew~\citeyearpar{Morales1996} have devised a notation to describe the members in a certain family of $d$-tupling functions for the non-negative integers.  The functions in this family all have diagonal shells.  Using Morales and Lew's notation,
\begin{equation*}
s_d(x_1,x_2,\ldots,x_d)=AAA\cdots A(x_d,x_{d-1},\ldots,x_1)
\end{equation*}
and
\begin{equation*}
\chi_d(x_1,x_2,\ldots,x_d)=BAA\cdots A(x_d,x_{d-1},\ldots,x_1)
\end{equation*}
for all integers $d>2$.
}
%
%

A $d$-tupling function with the shell numbering $\max(x_1,x_2,\ldots,x_d)$ is said to have \emph{cubic shells}.  (Although in the $d=2$ case, the term \emph{square shells} is sometimes used, instead.)  The Rosenberg-Strong pairing function $r_2(x,y)$ has cubic shells.  Other pairing functions with cubic shells have been described by P\'{e}ter~\citeyearpar[Sect. 1.27]{Peter1951} and Rosenberg~\citeyearpar[p. 294]{Rosenberg1978}.  In addition to cubic shells and diagonal shells, Rosenberg has also studied $d$-tupling functions with \emph{hyperbolic shells}~\citep{Rosenberg1975b}.

As another example of a popular pairing function, let $q\colon\mathbb{N}^2\to\mathbb{N}$ be defined so that
\begin{equation*}
q(x,y)=2^y(2x+1)-1.
\end{equation*}
Variants of this pairing function have been used by authors in computer science~\citep{Minsky1967,Cutland1980,Davis1983} and set theory~\citep{Sierpinski1912,Hausdorff1927,Hrbacek1978}, presumably because the inverse $q^{-1}(z)$ is easily calculated from the binary representation of $z+1$.  In this context, it is convenient to define the shell number of $(x,y)$ to be the number of bits in the binary representation of $q(x,y)+1$.  Notice that $\bigl\lceil\,\log_2(x+1)\bigr\rceil$ is the number of bits in the binary representation of the non-negative integer $x$, where $\lceil t\rceil$ denotes the ceiling of $t$ for each real number $t$.  Therefore, the shell number of $(x,y)$ is given by the formula
\begin{equation*}
y+1+\bigl\lceil\,\log_2(x+1)\bigr\rceil.
\end{equation*}
The shells of the pairing function $q(x,y)$ are illustrated in 
Figure~\ref{q-illust}.
\begin{figure}
\begin{center}
\begin{picture}(311,171)(-29,-32)
%
%
\multiput(0,0)(0,36){4}{\multiput(0,0)(36,0){8}{\circle*{4}}}
%
\thinlines
\put(-18,-18){\line(1,0){288}}
\put(-18,-18){\line(0,1){144}}
%
\put(276,-18){\makebox(0,0)[l]{$x$}}
\put(-18,132){\makebox(0,0)[b]{$y$}}
%
\thinlines
\multiput(0,-18)(36,0){8}{\line(0,1){3}}
\multiput(-18,0)(0,36){4}{\line(1,0){3}}
%
\put(0,-25){\makebox(0,0)[t]{0}}
\put(36,-25){\makebox(0,0)[t]{1}}
\put(72,-25){\makebox(0,0)[t]{2}}
\put(108,-25){\makebox(0,0)[t]{3}}
\put(144,-25){\makebox(0,0)[t]{4}}
\put(180,-25){\makebox(0,0)[t]{5}}
\put(216,-25){\makebox(0,0)[t]{6}}
\put(252,-25){\makebox(0,0)[t]{7}}
%
\put(-24,0){\makebox(0,0)[r]{0}}
\put(-24,36){\makebox(0,0)[r]{1}}
\put(-24,72){\makebox(0,0)[r]{2}}
\put(-24,108){\makebox(0,0)[r]{3}}
%
\multiput(4.293,31.707)(0,36){3}{\multiput(0,0)(2.109,-2.109){14}{\circle*{1}}}
\multiput(40.293,31.707)(0,36){2}{\multiput(0,0)(2.109,-2.109){14}{\circle*{1}}}
\multiput(78.071,0)(0,36){2}{\multiput(0,0)(2.982,0){9}{\circle*{1}}}
\multiput(112.293,31.707)(2.109,-2.109){14}{\circle*{1}}
\multiput(150.071,0)(36,0){3}{\multiput(0,0)(2.982,0){9}{\circle*{1}}}
%
\put(4,5){\makebox(0,0)[bl]{0}}
\put(4,41){\makebox(0,0)[bl]{1}}
\put(40,5){\makebox(0,0)[bl]{2}}
\put(4,77){\makebox(0,0)[bl]{3}}
\put(76,5){\makebox(0,0)[bl]{4}}
\put(40,41){\makebox(0,0)[bl]{5}}
\put(112,5){\makebox(0,0)[bl]{6}}
\put(4,113){\makebox(0,0)[bl]{7}}
\put(148,5){\makebox(0,0)[bl]{8}}
\put(76,41){\makebox(0,0)[bl]{9}}
\put(184,5){\makebox(0,0)[bl]{10}}
\put(40,77){\makebox(0,0)[bl]{11}}
\put(220,5){\makebox(0,0)[bl]{12}}
\put(112,41){\makebox(0,0)[bl]{13}}
\put(256,5){\makebox(0,0)[bl]{14}}
\end{picture}
\caption{The pairing function $q(x,y)$.  Points connected by a sequence of dotted line segments have the same shell number.}
\label{q-illust}
\end{center}
\end{figure}
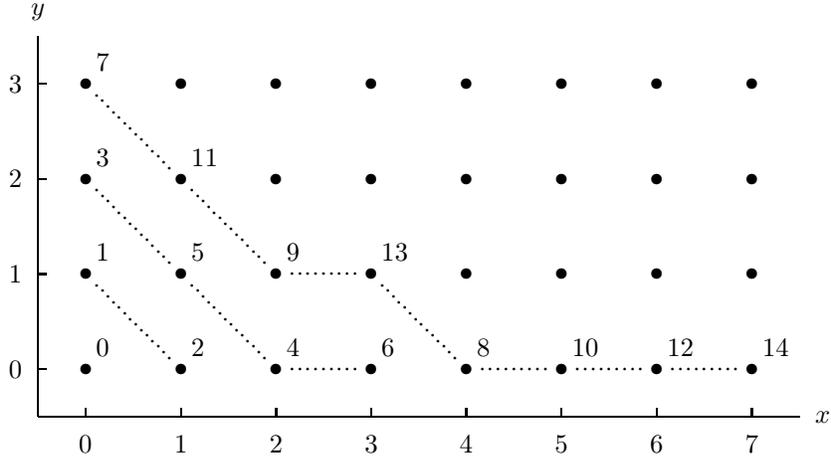

Regan~\citeyearpar{Regan1992} has investigated the computational complexity of pairing functions for the positive integers, and has described several pairing functions with low computational complexity.  Among Regan's results is a pairing function that can be computed in linear time and constant space.

\section{Additional results}\label{additional}

There are several different ways that the shell numberings of a $d$-tupling function can be characterized.  In particular, we have the following lemma.

\begin{lemma}\label{alt-shell-numbering}
Let $f\colon\mathbb{N}^d\to\mathbb{N}$ be any $d$-tupling function, and let $\sigma\colon\mathbb{N}^d\to\mathbb{N}$ be any function.  The following statements are equivalent.
\begin{enumerate}
\item[(a)] $\sigma$ is a shell numbering for $f$.
\item[(b)] For all $i,j\in\mathbb{N}$,
\begin{equation*}
\sigma\bigl(f^{-1}(i)\bigr)<\sigma\bigl(f^{-1}(j)\bigr)\quad\text{implies}\quad i<j.
\end{equation*}
\item[(c)] The sequence
\begin{equation*}
\sigma\bigl(f^{-1}(0)\bigr),\;\;\sigma\bigl(f^{-1}(1)\bigr),\;\;\sigma\bigl(f^{-1}(2)\bigr),\;\;\ldots
\end{equation*}
is non-decreasing.
\end{enumerate}
\end{lemma}
\begin{proof}
By Definition~\ref{shell-numbering}, $\sigma$ is a shell numbering for $f$ if and only if
\begin{equation*}
\sigma(\mathbf{x})<\sigma(\mathbf{y})\quad\text{implies}\quad f(\mathbf{x})<f(\mathbf{y})
\end{equation*}
for all $\mathbf{x},\mathbf{y}\in\mathbb{N}^d$.  But $f^{-1}$ is a one-to-one correspondence from $\mathbb{N}$ to $\mathbb{N}^d$.  Hence, letting $\mathbf{x}=f^{-1}(i)$ and $\mathbf{y}=f^{-1}(j)$, $\sigma$ is a shell numbering for $f$ if and only if 
\begin{equation*}
\sigma\bigl(f^{-1}(i)\bigr)<\sigma\bigl(f^{-1}(j)\bigr)\quad\text{implies}\quad i<j
\end{equation*}
for all $i,j\in\mathbb{N}$.  We have shown that (a) is equivalent to (b).  And taking the contrapositive, (b) is equivalent to the statement that
\begin{equation*}
i\geq j\quad\text{implies}\quad\sigma\bigl(f^{-1}(i)\bigr)\geq\sigma\bigl(f^{-1}(j)\bigr)
\end{equation*}
for all $i,j\in\mathbb{N}$.  From the definition of a non-decreasing function, it immediately follows that (b) is equivalent to (c).
\end{proof}

For each function $\sigma\colon\mathbb{N}^d\to\mathbb{N}$ and each non-negative integer $n$, let $U_\sigma^{<n}$ denote the set of all points $\mathbf{y}\in\mathbb{N}^d$ such that $\sigma(\mathbf{y})<n$.  The following theorem provides another useful way to characterize shell numberings.

\begin{theorem}\label{shell-inequality}
Let $f\colon\mathbb{N}^d\to\mathbb{N}$ be any $d$-tupling function.  A function $\sigma\colon\mathbb{N}^d\to\mathbb{N}$ is a shell numbering for $f$ if and only if, for all points $\mathbf{x}$ in $\mathbb{N}^d$,
\begin{equation}\label{U-ineq}
\bigl\lvert U_\sigma^{<n}\bigr\rvert\leq f(\mathbf{x})<\bigl\lvert U_\sigma^{<n+1}\bigr\rvert,
\end{equation}
where $n=\sigma(\mathbf{x})$.
\end{theorem}

\begin{remark*}
In the statement of this theorem, if the cardinality of $U_\sigma^{<n}$ is infinite, then the inequality $\bigl\lvert U_\sigma^{<n}\bigr\rvert\leq i$ is false for all non-negative integers $i$.  If the cardinality of $U_\sigma^{<n+1}$ is infinite, then $i<\bigl\lvert U_\sigma^{<n+1}\bigr\rvert$ is true for all non-negative integers $i$.
\end{remark*}

\begin{proof}
Let $\sigma$ be a function from $\mathbb{N}^d$ to $\mathbb{N}$.  Note that if $\sigma(\mathbf{x})<\sigma(\mathbf{y})$ for any $\mathbf{x},\mathbf{y}\in\mathbb{N}^d$, then inequality~\eqref{U-ineq} implies that
\begin{equation*}
f(\mathbf{x})<\bigl\lvert U_\sigma^{<\sigma(\mathbf{x})+1}\bigr\rvert\leq\bigl\lvert U_\sigma^{<\sigma(\mathbf{y})}\bigr\rvert\leq f(\mathbf{y}).
\end{equation*}
Hence, $\sigma$ is a shell numbering for $f$ if inequality~\eqref{U-ineq} holds.

Conversely, suppose that $\sigma$ is a shell numbering for $f$, and consider any $\mathbf{x}\in\mathbb{N}^d$.  Let $i=f(\mathbf{x})$ and $n=\sigma(\mathbf{x})$, and note that
\begin{equation*}
\sigma\bigl(f^{-1}(i)\bigr)=\sigma\bigl(f^{-1}\bigl(f(\mathbf{x})\bigr)\bigr)=\sigma(\mathbf{x})=n.
\end{equation*}
Now, by Lemma~\ref{alt-shell-numbering},
\begin{equation*}
\sigma\bigl(f^{-1}(0)\bigr),\;\;\sigma\bigl(f^{-1}(1)\bigr),\;\;\sigma\bigl(f^{-1}(2)\bigr),\;\;\ldots
\end{equation*}
is a non-decreasing sequence.  Hence, it must be the case that
\begin{equation*}
U_\sigma^{<n}\subseteq\Bigl\{\,f^{-1}(0),\,f^{-1}(1),\,\ldots,\,f^{-1}(i-1)\,\Bigr\}
\end{equation*}
because $\sigma\bigl(f^{-1}(i)\bigr)=n$ and $U_\sigma^{<n}$ is the set of all points $\mathbf{y}\in\mathbb{N}^d$ such that $\sigma(\mathbf{y})<n$.  Similarly,
\begin{equation*}
\Bigl\{\,f^{-1}(0),\,f^{-1}(1),\,\ldots,\,f^{-1}(i-1),\,f^{-1}(i)\,\Bigr\}\subseteq U_\sigma^{<n+1}.
\end{equation*}
We may conclude that $\bigl\lvert U_\sigma^{<n}\bigr\rvert\leq i<\bigl\lvert U_\sigma^{<n+1}\bigr\rvert$.  That is, $\bigl\lvert U_\sigma^{<n}\bigr\rvert\leq f(\mathbf{x})<\bigl\lvert U_\sigma^{<n+1}\bigr\rvert$.
\end{proof}

Given a function $\sigma\colon\mathbb{N}^d\to\mathbb{N}$ and a non-negative integer $n$, $\bigl\lvert U_\sigma^{<n}\bigr\rvert$ is the number of points $\mathbf{x}$ such that $\sigma(\mathbf{x})<n$.  From this fact, simple combinatorial arguments are often sufficient to calculate $\bigl\lvert U_\sigma^{<n}\bigr\rvert$.  For example, using the function $\delta(x_1,x_2,\ldots,x_d)=x_1+x_2+\cdots+x_d$, we have that $\bigl\lvert U_\delta^{<w}\bigr\rvert=\binom{w+d-1}{d}$.  And using $\max(x_1,x_2,\ldots,x_d)$, we have that $\bigl\lvert U_{\max}^{<m}\bigr\rvert=m^d$.  These two observations imply the following corollaries of Theorem~\ref{shell-inequality}.

\begin{corollary}\label{diag-inequality}
Let $f\colon\mathbb{N}^d\to\mathbb{N}$ be any $d$-tupling function.  The function $f$ has diagonal shells if and only if, for all points $(x_1,x_2,\ldots,x_d)$ in $\mathbb{N}^d$,
\begin{equation*}
\tbinom{w+d-1}{d}\leq f(x_1,x_2,\ldots,x_d)<\tbinom{w+d}{d},
\end{equation*}
where $w=x_1+x_2+\cdots+x_d$.
\end{corollary}

\begin{corollary}\label{cubic-inequality}
Let $f\colon\mathbb{N}^d\to\mathbb{N}$ be any $d$-tupling function.  The function $f$ has cubic shells if and only if, for all points $(x_1,x_2,\ldots,x_d)$ in $\mathbb{N}^d$,
\begin{equation*}
m^d\leq f(x_1,x_2,\ldots,x_d)<(m+1)^d,
\end{equation*}
where $m=\max(x_1,x_2,\ldots,x_d)$.
\end{corollary}

Next, we say that a function $f\colon\mathbb{N}^d\to\mathbb{N}$ is \emph{max-dominating} if and only if
\begin{equation*}
\max(\mathbf{x})\leq f(\mathbf{x})\quad\text{for all points $\mathbf{x}\in\mathbb{N}^d$}.
\end{equation*}
%
%
In particular, if $f$ is max-dominating then
\begin{equation*}
\mathbf{x}\ne(0,0,\ldots,0)\quad\text{implies}\quad0<\max(\mathbf{x})\leq f(\mathbf{x}).
\end{equation*}
It immediately follows that $f(0,0,\ldots,0)=0$ for every max-dominating $d$-tupling function $f\colon\mathbb{N}^d\to\mathbb{N}$, since $f(\mathbf{x})$ must be $0$ for some $\mathbf{x}\in\mathbb{N}^d$.
%
%
We also have the following lemma.

\begin{lemma}\label{max-dominating}
If a $d$-tupling function $f\colon\mathbb{N}^d\to\mathbb{N}$ has diagonal shells or cubic shells, then $f$ is max-dominating.
\end{lemma}
\begin{proof}
Suppose that $f\colon\mathbb{N}^d\to\mathbb{N}$ is a $d$-tupling function, and consider any non-negative integers $x_1$, $x_2$, \ldots, $x_d$.  Note that if $f$ has diagonal shells, then by Corollary~\ref{diag-inequality},
\begin{equation*}
\max(x_1,x_2,\ldots,x_d)\leq x_1+\cdots+x_d\leq\tbinom{x_1+\cdots+x_d+d-1}{d}\leq f(x_1,x_2,\ldots,x_d).
\end{equation*}
Alternatively, if $f$ has cubic shells, then
\begin{equation*}
\max(x_1,x_2,\ldots,x_d)\leq\bigl(\max(x_1,x_2,\ldots,x_d)\bigr)^d\leq f(x_1,x_2,\ldots,x_d)
\end{equation*}
by Corollary~\ref{cubic-inequality}.  In either case, $f$ is max-dominating.
\end{proof}

The Rosenberg-Strong pairing function, described in the previous section, is the most well-known example of a pairing function with cubic shells.  A higher-dimensional generalization of this function has also been introduced by Rosenberg and Strong~\citetext{\citeyear{Rosenberg1972}; \citealp{Rosenberg1974}}.  In particular, the \emph{Rosenberg-Strong $d$-tupling function}\footnote{
The function originally described by Rosenberg and Strong was a $d$-tupling function for the positive integers.  The variant defined here is obtained by translating the Rosenberg-Strong function from the positive integers to the non-negative integers, and by reversing the order of the arguments.
} for the non-negative integers is defined recursively so that $r_1(x_1)=x_1$, and so that for all integers $d>1$,
\begin{equation*}
r_d(x_1,\ldots,x_{d-1},x_d)=r_{d-1}(x_1,\ldots,x_{d-1})+m^d+(m-x_d)\bigl((m+1)^{d-1}-m^{d-1}\bigr),
\end{equation*}
where $m=\max(x_1,\ldots,x_{d-1},x_d)$.  Note that equation~\eqref{r2} agrees with this definition when $d=2$.  It follows from Corollary~\ref{cubic-inequality} and Lemma~\ref{r-cubic} that $r_d(x_1,x_2,\ldots,x_d)$ has cubic shells for each positive integer $d$.

The inverse of the Rosenberg-Strong $d$-tupling function $r_d(x_1,x_2,\ldots,x_d)$ can also be defined recursively.  In particular, $r_1^{-1}(z)=z$.  And for all integers $d>1$,
\begin{equation}\label{r-inv}
r_d^{-1}(z)=\Bigl(r_{d-1}^{-1}\bigl(z-m^d-(m-x_d)((m+1)^{d-1}-m^{d-1})\bigr),x_d\Bigr),
\end{equation}
where
\begin{equation}\label{xd}
x_d=m-\Biggl\lfloor\frac{\max\bigl(0,z-m^d-m^{d-1}\bigr)}{(m+1)^{d-1}-m^{d-1\rule{0pt}{6pt}}}\Biggr\rfloor
\end{equation}
and  $m=\bigl\lfloor\sqrt[d]{z}\,\bigr\rfloor$.

\section{Applications}

It is a common convention~\citep{West1996} that the vertices in a binary tree may have $0$, $1$, or $2$ children.  A binary tree where each vertex has either $0$ or $2$ children is said to be a \emph{full binary tree}.  We use $o$ to denote the trivial binary tree that has only one vertex, and we use $\tau(a,b)$ to denote the binary tree with left subtree $a$ and right subtree $b$.  Then, the set $T$ of all full binary trees is the smallest set that contains $o$ and that is closed under the operation $\tau$.  Let the \emph{height} of a binary tree be the length of the longest path from a leaf to the tree's root.  We use $H(t)$ to denote the height of the binary tree $t$.  For all binary trees of the form $\tau(a,b)$,
\begin{equation*}
H\bigl(\tau(a,b)\bigr)=1+\max\bigl(H(a),H(b)\bigr).
\end{equation*}
The only binary tree of height zero is the trivial binary tree $o$.

An important application of pairing functions is in the enumeration of full binary trees.

\begin{theorem}\label{tree-enum}
Let $f$ be any max-dominating pairing function for the non-negative integers.  Let $\phi_f(0)=o$, and for each pair $(x,y)$ of non-negative integers let
\begin{equation*}
\phi_f\bigl(f(x,y)+1\bigr)=\tau\bigl(\phi_f(x),\phi_f(y)\bigr).
\end{equation*}
Then, $\phi_f$ is an enumeration of the set $T$ of full binary trees.
\end{theorem}
\begin{proof}
First, we show that $\phi_f(n)\in T$ for each $n\in\mathbb{N}$.  The proof is by strong induction on $n$.  For the base step, notice that $\phi_f(0)=o\in T$.  Now consider any non-negative integer $n$ and suppose, as the induction hypothesis, that $\phi_f(i)\in T$ for all non-negative integers $i\leq n$.  Since $f\colon\mathbb{N}^2\to\mathbb{N}$ is a pairing function, there exist non-negative integers $x$ and $y$ such that $f(x,y)=n$.  Then, by the definition of $\phi_f$,
\begin{equation}\label{phi-induct}
\phi_f(n+1)=\phi_f\bigl(f(x,y)+1\bigr)=\tau\bigl(\phi_f(x),\phi_f(y)\bigr).
\end{equation}
But $x\leq f(x,y)=n$ and $y\leq f(x,y)=n$ because $f$ is max-dominating.  Therefore, by the induction hypothesis, $\phi_f(x)\in T$ and $\phi_f(y)\in T$.  We may conclude by equation~\eqref{phi-induct} that $\phi_f(n+1)\in T$.

Next, we show that for each $t\in T$ there exists a unique $n\in\mathbb{N}$ such that $\phi_f(n)=t$.  The proof is by strong induction on the height of $t$.  For the base step, notice that $0$ is the only non-negative integer $n$ such that $\phi_f(n)=o$.  Then consider any non-negative integer $k$ and suppose, as the induction hypothesis, that for each $t\in T$ whose height is less than or equal to $k$, there exists a unique $n\in\mathbb{N}$ such that $\phi_f(n)=t$.  Now consider any $t\in T$ with height $k+1$.  Since the height of $t$ is greater than zero, it must be the case that $t=\tau(a,b)$ for some binary trees $a\in T$ and $b\in T$ whose heights are each less than or equal to $k$.  By the induction hypothesis, there exists a unique $x\in\mathbb{N}$ such that $\phi_f(x)=a$ and a unique $y\in\mathbb{N}$ such that $\phi_f(y)=b$.  Therefore, $n=f(x,y)+1$ is the unique non-negative integer such that
\begin{equation*}
\phi_f(n)=\phi_f\bigl(f(x,y)+1\bigr)=\tau\bigl(\phi_f(x),\phi_f(y)\bigr)=\tau(a,b)=t.
\end{equation*}
We may conclude that $\phi_f$ is a one-to-one correspondence from $\mathbb{N}$ to $T$.  That is, $\phi_f$ is an enumeration of the set $T$.
\end{proof}

It follows that any max-dominating pairing function for the non-negative integers can be used to construct an enumeration of the full binary trees.\footnote{
Incidentally, Theorem~\ref{tree-enum} can also be used to construct an enumeration of \emph{all} binary trees, including those binary trees that are not full binary trees.  Let $D$ denote the operation of \emph{defoliation}.  That is, for each non-trivial binary tree $t$, let $D(t)$ denote the tree that is obtained from $t$ by deleting all of its leaves.  The function $D$ is a one-to-one correspondence from the set of non-trivial full binary trees to the set of all binary trees.  As a consequence, $D\bigl(\phi_f(n+1)\bigr)$ is an enumeration of all binary trees, where $f$ is any max-dominating pairing function for the non-negative integers.
}  For example, the enumeration $\phi_c$ that is constructed using Cantor's pairing function is illustrated in Figure~\ref{c-tree},
\begin{figure}
\begin{center}
\begin{picture}(263,173)(-16.5,-173)
%
%
\thinlines
\put(0,-7){\makebox(0,0)[b]{0}}
\put(0,-12){\treeo}
\put(45,-7){\makebox(0,0)[b]{1}}
\put(45,-12){\treei{\treeo}{\treeo}}
\put(90,-7){\makebox(0,0)[b]{2}}
\put(90,-12){\treei{\treeo}{\treeii{\treeo}{\treeo}}}
\put(135,-7){\makebox(0,0)[b]{3}}
\put(135,-12){\treei{\treeii{\treeo}{\treeo}}{\treeo}}
\put(180,-7){\makebox(0,0)[b]{4}}
\put(180,-12){\treei{\treeo}{\treeii{\treeo}{\treeiii{\treeo}{\treeo}}}}
\put(225,-7){\makebox(0,0)[b]{5}}
\put(225,-12){\treei{\treeii{\treeo}{\treeo}}{\treeii{\treeo}{\treeo}}}
\put(0,-62){\makebox(0,0)[b]{6}}
\put(0,-67){\treei{\treeii{\treeo}{\treeiii{\treeo}{\treeo}}}{\treeo}}
\put(45,-62){\makebox(0,0)[b]{7}}
\put(45,-67){\treei{\treeo}{\treeii{\treeiii{\treeo}{\treeo}}{\treeo}}}
\put(90,-62){\makebox(0,0)[b]{8}}
\put(90,-67){\treei{\treeii{\treeo}{\treeo}}{\treeii{\treeo}{\treeiii{\treeo}{\treeo}}}}
\put(135,-62){\makebox(0,0)[b]{9}}
\put(135,-67){\treei{\treeii{\treeo}{\treeiii{\treeo}{\treeo}}}{\treeii{\treeo}{\treeo}}}
\put(180,-62){\makebox(0,0)[b]{10}}
\put(180,-67){\treei{\treeii{\treeiii{\treeo}{\treeo}}{\treeo}}{\treeo}}
\put(225,-62){\makebox(0,0)[b]{11}}
\put(225,-67){\treei{\treeo}{\treeii{\treeo}{\treeiii{\treeo}{\treeiii{\treeo}{\treeo}}}}}
\put(0,-127){\makebox(0,0)[b]{12}}
\put(0,-132){\treei{\treeii{\treeo}{\treeo}}{\treeii{\treeiii{\treeo}{\treeo}}{\treeo}}}
\put(45,-127){\makebox(0,0)[b]{13}}
\put(45,-132){\treei{\treeii{\treeo}{\treeiii{\treeo}{\treeo}}}{\treeii{\treeo}{\treeiii{\treeo}{\treeo}}}}
\put(90,-127){\makebox(0,0)[b]{14}}
\put(90,-132){\treei{\treeii{\treeiii{\treeo}{\treeo}}{\treeo}}{\treeii{\treeo}{\treeo}}}
\put(135,-127){\makebox(0,0)[b]{15}}
\put(135,-132){\treei{\treeii{\treeo}{\treeiii{\treeo}{\treeiii{\treeo}{\treeo}}}}{\treeo}}
\put(180,-127){\makebox(0,0)[b]{16}}
\put(180,-132){\treei{\treeo}{\treeii{\treeiii{\treeo}{\treeo}}{\treeiii{\treeo}{\treeo}}}}
\put(225,-127){\makebox(0,0)[b]{17}}
\put(225,-132){\treei{\treeii{\treeo}{\treeo}}{\treeii{\treeo}{\treeiii{\treeo}{\treeiii{\treeo}{\treeo}}}}}
\end{picture}
\caption{The first $18$ full binary trees in the enumeration $\phi_c$.}
\label{c-tree}
\end{center}
\end{figure}
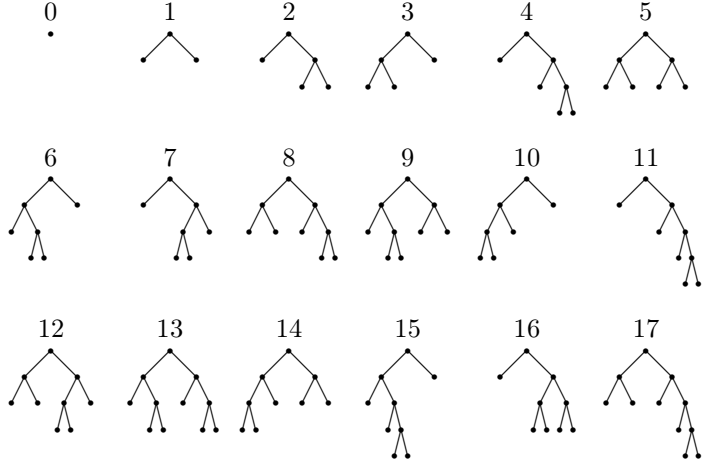
and the enumeration $\phi_{r_2}$, constructed using the Rosenberg-Strong pairing function, is illustrated in Figure~\ref{r-tree}.
\begin{figure}[t]
\begin{center}
\begin{picture}(260.5,153)(-16.5,-153)
%
%
\thinlines
\put(0,-7){\makebox(0,0)[b]{0}}
\put(0,-12){\treeo}
\put(45,-7){\makebox(0,0)[b]{1}}
\put(45,-12){\treei{\treeo}{\treeo}}
\put(90,-7){\makebox(0,0)[b]{2}}
\put(90,-12){\treei{\treeo}{\treeii{\treeo}{\treeo}}}
\put(135,-7){\makebox(0,0)[b]{3}}
\put(135,-12){\treei{\treeii{\treeo}{\treeo}}{\treeii{\treeo}{\treeo}}}
\put(180,-7){\makebox(0,0)[b]{4}}
\put(180,-12){\treei{\treeii{\treeo}{\treeo}}{\treeo}}
\put(225,-7){\makebox(0,0)[b]{5}}
\put(225,-12){\treei{\treeo}{\treeii{\treeo}{\treeiii{\treeo}{\treeo}}}}
\put(0,-62){\makebox(0,0)[b]{6}}
\put(0,-67){\treei{\treeii{\treeo}{\treeo}}{\treeii{\treeo}{\treeiii{\treeo}{\treeo}}}}
\put(45,-62){\makebox(0,0)[b]{7}}
\put(45,-67){\treei{\treeii{\treeo}{\treeiii{\treeo}{\treeo}}}{\treeii{\treeo}{\treeiii{\treeo}{\treeo}}}}
\put(90,-62){\makebox(0,0)[b]{8}}
\put(90,-67){\treei{\treeii{\treeo}{\treeiii{\treeo}{\treeo}}}{\treeii{\treeo}{\treeo}}}
\put(135,-62){\makebox(0,0)[b]{9}}
\put(135,-67){\treei{\treeii{\treeo}{\treeiii{\treeo}{\treeo}}}{\treeo}}
\put(180,-62){\makebox(0,0)[b]{10}}
\put(180,-67){\treei{\treeo}{\treeii{\treeiii{\treeo}{\treeo}}{\treeiii{\treeo}{\treeo}}}}
\put(225,-62){\makebox(0,0)[b]{11}}
\put(225,-67){\treei{\treeii{\treeo}{\treeo}}{\treeii{\treeiii{\treeo}{\treeo}}{\treeiii{\treeo}{\treeo}}}}
\put(0,-117){\makebox(0,0)[b]{12}}
\put(0,-122){\treei{\treeii{\treeo}{\treeiii{\treeo}{\treeo}}}{\treeii{\treeiii{\treeo}{\treeo}}{\treeiii{\treeo}{\treeo}}}}
\put(45,-117){\makebox(0,0)[b]{13}}
\put(45,-122){\treei{\treeii{\treeiii{\treeo}{\treeo}}{\treeiii{\treeo}{\treeo}}}{\treeii{\treeiii{\treeo}{\treeo}}{\treeiii{\treeo}{\treeo}}}}
\put(90,-117){\makebox(0,0)[b]{14}}
\put(90,-122){\treei{\treeii{\treeiii{\treeo}{\treeo}}{\treeiii{\treeo}{\treeo}}}{\treeii{\treeo}{\treeiii{\treeo}{\treeo}}}}
\put(135,-117){\makebox(0,0)[b]{15}}
\put(135,-122){\treei{\treeii{\treeiii{\treeo}{\treeo}}{\treeiii{\treeo}{\treeo}}}{\treeii{\treeo}{\treeo}}}
\put(180,-117){\makebox(0,0)[b]{16}}
\put(180,-122){\treei{\treeii{\treeiii{\treeo}{\treeo}}{\treeiii{\treeo}{\treeo}}}{\treeo}}
\put(225,-117){\makebox(0,0)[b]{17}}
\put(225,-122){\treei{\treeo}{\treeii{\treeiii{\treeo}{\treeo}}{\treeo}}}
\end{picture}
\caption{The first $18$ full binary trees in the enumeration $\phi_{r_2}$.}
\label{r-tree}
\end{center}
\end{figure}
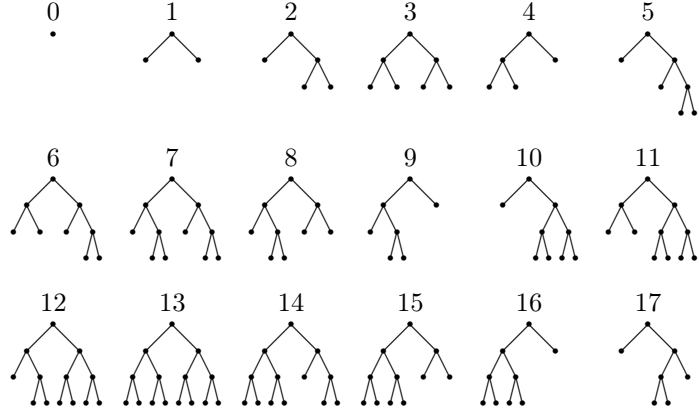

In these illustrations, one advantage of the Rosenberg-Strong pairing function over Cantor's pairing function is apparent: the heights of trees never decrease in the enumeration $\phi_{r_2}$.  In fact, we have the following theorem.

\begin{theorem}
Let $f\colon\mathbb{N}^2\to\mathbb{N}$ be any pairing function with cubic shells.  Then, the sequence
\begin{equation*}
H\bigl(\phi_f(0)\bigr),\;\;H\bigl(\phi_f(1)\bigr),\;\;H\bigl(\phi_f(2)\bigr),\;\;\ldots
\end{equation*}
is non-decreasing.
\end{theorem}
\begin{proof}
By definition, the sequence is non-decreasing if and only if, for all $i,j\in\mathbb{N}$,
\begin{equation}\label{h-shell}
i\geq j\quad\text{implies}\quad H\bigl(\phi_f(i)\bigr)\geq H\bigl(\phi_f(j)\bigr).
\end{equation}
We prove this statement by strong induction on $i$.  For the base step, note that condition~\eqref{h-shell} is true for all $j\in\mathbb{N}$ if $i=0$.  Now consider any non-negative integer $n$ and suppose, as the induction hypothesis, that condition~\eqref{h-shell} is true for all $i,j\in\mathbb{N}$ such that $i\leq n$.  Next, consider any positive integer $j\leq n+1$, and let $(x,y)=f^{-1}(j-1)$.  By the definition of $\phi_f$,
\begin{equation*}
\phi_f(j)=\phi_f\bigl(f(x,y)+1\bigr)=\tau\bigl(\phi_f(x),\phi_f(y)\bigr).
\end{equation*}
Therefore,
\begin{equation}\label{h-max}
H\bigl(\phi_f(j)\bigr)=1+\max\bigl(H(\phi_f(x)),H(\phi_f(y))\bigr).
\end{equation}
But by Lemma~\ref{max-dominating}, $f$ is max-dominating.  So,
\begin{equation*}
x\leq f(x,y)=j-1\leq n\quad\text{and}\quad y\leq f(x,y)=j-1\leq n.
\end{equation*}
It immediately follows from the induction hypothesis that if $x\geq y$ then
%
%
$H\bigl(\phi_f(x)\bigr)\linebreak[0]\geq H\bigl(\phi_f(y)\bigr)$.  Similarly, if $y\geq x$ then $H\bigl(\phi_f(y)\bigr)\geq H\bigl(\phi_f(x)\bigr)$.  In either case,
\begin{equation*}
\max\bigl(H(\phi_f(x)),H(\phi_f(y))\bigr)=H\bigl(\phi_f(\max(x,y))\bigr).
\end{equation*}
Therefore, by equation~\eqref{h-max},
\begin{equation*}
H\bigl(\phi_f(j)\bigr)=1+H\bigl(\phi_f(\max(x,y))\bigr).
\end{equation*}
And since this is true for all positive integers $j\leq n+1$, it is true for $n+1$.  Hence,
\begin{equation*}
H\bigl(\phi_f(n+1)\bigr)=1+H\bigl(\phi_f(\max(u,v))\bigr),
\end{equation*}
where $(u,v)=f^{-1}(n)$.  Furthermore,
\begin{equation*}
\max(u,v)<\max(x,y)\quad\text{implies}\quad f(u,v)<f(x,y)
\end{equation*}
because $f$ has cubic shells.  But $f(u,v)=n\geq j-1=f(x,y)$, so it must be the case that $\max(u,v)\geq\max(x,y)$.  And $\max(u,v)\leq f(u,v)=n$ because $f$ is max-dominating.  So, by the induction hypothesis,
\begin{align*}
H\bigl(\phi_f(\max(u,v))\bigr)&\geq H\bigl(\phi_f(\max(x,y))\bigr),\\
1+H\bigl(\phi_f(\max(u,v))\bigr)&\geq 1+H\bigl(\phi_f(\max(x,y))\bigr),\\
H\bigl(\phi_f(n+1)\bigr)&\geq H\bigl(\phi_f(j)\bigr).
\end{align*}
This is also true if $j=0$, since $H\bigl(\phi_f(0)\bigr)=0$.  Hence, we have shown that
\begin{equation*}
n+1\geq j\quad\text{implies}\quad H\bigl(\phi_f(n+1)\bigr)\geq H\bigl(\phi_f(j)\bigr)
\end{equation*}
for all $j\in\mathbb{N}$.
\end{proof}

Another important application, closely-related to the enumeration of full binary trees, is the enumeration of finite-length sequences.
%
%
For each positive integer $d$, the members of $\mathbb{N}^d$ are said to be the \emph{length-$d$ sequences} of non-negative integers.  Then,
\begin{equation*}
\mathbb{N}^*=\{()\}\cup\mathbb{N}^1\cup\mathbb{N}^2\cup\mathbb{N}^3\cup\cdots
\end{equation*}
is the set of all finite-length sequences of non-negative integers, where $()$ denotes the \emph{empty sequence} of length zero.
%
%
By convention, sequences of non-negative integers that have different lengths are distinct members of $\mathbb{N}^*$.

One simple way to construct an enumeration of $\mathbb{N}^*$ is to choose a pairing function $f\colon\mathbb{N}^2\to\mathbb{N}$, and to choose a $d$-tupling function $g_d\colon\mathbb{N}^d\to\mathbb{N}$ for each positive integer $d$.  Then, the function $\zeta_{f,g}$ that is defined so that $\zeta_{f,g}(0)=()$, and so that
\begin{equation*}
\zeta_{f,g}\bigl(f(x,y)+1\bigr)=g_{y+1}^{-1}(x)
\end{equation*}
for each pair $(x,y)$ of non-negative integers, is an enumeration of $\mathbb{N}^*$.
%
%

Another way to construct an enumeration of $\mathbb{N}^*$ is described by the following theorem.  The proof of this theorem is similar to the proof of Theorem~\ref{tree-enum}.

\begin{theorem}
Let $f$ be any max-dominating pairing function for the non-negative integers.
%
%
Let $\xi_f(0)=()$, and for each pair $(x,y)$ of non-negative integers let
\begin{equation*}
\xi_f\bigl(f(x,y)+1\bigr)=\begin{cases}
y &\text{if $x=0$}\\
\bigl(\xi_f(x),y\bigr) &\text{otherwise}\rule{0pt}{12pt}
\end{cases}.
\end{equation*}
Then, $\xi_f$ is an enumeration of $\mathbb{N}^*$.
\end{theorem}
\begin{proof}
First, we show that $\xi_f(n)\in\mathbb{N}^*$ for each $n\in\mathbb{N}$.  Since $\xi_f(0)=()\in\mathbb{N}^*$, it suffices to prove that for each pair $(x,y)$ of non-negative integers there exists a positive integer $d$ such that $\xi_f\bigl(f(x,y)+1\bigr)\in\mathbb{N}^d$.  The proof is by strong induction on $x$.  For the base step, notice that
\begin{equation*}
\xi_f\bigl(f(0,y)+1\bigr)=y\in\mathbb{N}^1
\end{equation*}
for all non-negative integers $y$.  Now consider any non-negative integer $x$ and suppose, as the induction hypothesis, that for each non-negative integer $i\leq x$, and for each non-negative integer $y$, there exists a positive integer $d$ such that $\xi_f\bigl(f(i,y)+1\bigr)\in\mathbb{N}^d$.  Since $f\colon\mathbb{N}^2\to\mathbb{N}$ is a pairing function, there exist non-negative integers $i$ and $j$ such that $f(i,j)=x$.  Then, by the definition of $\xi_f$,
\begin{equation}\label{zeta-induct}
\xi_f\bigl(f(x+1,y)+1\bigr)=\Bigl(\xi_f(x+1),y\Bigr)=\Bigl(\xi_f\bigl(f(i,j)+1\bigr),y\Bigr)
\end{equation}
for all non-negative integers $y$.  But $i\leq f(i,j)=x$ because $f$ is max-dominating.  Therefore, by the induction hypothesis, $\xi_f\bigl(f(i,j)+1\bigr)\in\mathbb{N}^d$ for some positive integer $d$.  We may conclude by equation~\eqref{zeta-induct} that $\xi_f\bigl(f(x+1,y)+1\bigr)\in\mathbb{N}^{d+1}$ for all non-negative integers $y$.

Next, we show that for each $\mathbf{u}\in\mathbb{N}^*$ there exists a unique $n\in\mathbb{N}$ such that $\xi_f(n)=\mathbf{u}$.  The proof is by induction on the length of the sequence $\mathbf{u}$.  For the base step, notice that if $\mathbf{u}=()$ then $n=0$ is the unique non-negative integer such that $\xi_f(n)=\mathbf{u}$.  And if $\mathbf{u}\in\mathbb{N}^1$ then $n=f(0,\mathbf{u})+1$ is the unique non-negative integer such that $\xi_f(n)=\mathbf{u}$.  Next, consider any positive integer $d$ and suppose, as the induction hypothesis, that for each $\mathbf{u}\in\mathbb{N}^d$ there exists a unique $n\in\mathbb{N}$ such that $\xi_f(n)=\mathbf{u}$.  Now consider any $\mathbf{u}\in\mathbb{N}^{d+1}$.  It must be the case that $\mathbf{u}=(\mathbf{v},y)$ for some unique $\mathbf{v}\in\mathbb{N}^d$ and unique $y\in\mathbb{N}$.  By the induction hypothesis, there also exists a unique $x\in\mathbb{N}$ such that $\xi_f(x)=\mathbf{v}$.  Note that $x\ne0$ because $\mathbf{v}\ne()$.  Therefore, $n=f(x,y)+1$ is the unique non-negative integer such that
\begin{equation*}
\xi_f(n)=\xi_f\bigl(f(x,y)+1\bigr)=\bigl(\xi_f(x),y\bigr)=(\mathbf{v},y)=\mathbf{u}.
\end{equation*}
We may conclude that $\xi_f$ is a one-to-one correspondence from $\mathbb{N}$ to $\mathbb{N}^*$.  That is, $\xi_f$ is an enumeration of the set $\mathbb{N}^*$.
\end{proof}

The Rosenberg-Strong pairing function was originally
%
%
devised~\citetext{Rosenberg and Strong, \citeyear{Rosenberg1972}; \citealp{Rosenberg1974}} for applications to data storage in computer science.  One of its key features is that if the binary representations of $x$ and $y$ each have $n$ or fewer bits, then the binary representation of $r_2(x,y)$ has $2n$ or fewer bits.  This property is often useful when implementing the Rosenberg-Strong pairing function on a contemporary computer, and it is a consequence of the fact that, in Rosenberg's words, $r_2$ ``manages storage perfectly for square arrays''~\citep{Rosenberg1975b}.  In contrast, Cantor's pairing function $c(x,y)$ does not possess this property.  For example, the binary representation of $3=\bigl(11\bigr)_2$ has two bits, and the binary representation of $2=\bigl(10\bigr)_2$ also has two bits, but
\begin{equation*}
c(3,2)=18=\bigl(10010\bigr)_2
\end{equation*}
has more than 4 bits in its binary representation.

More generally, we have the following theorem.

\begin{theorem}\label{cubic-bits}
Let $f\colon\mathbb{N}^d\to\mathbb{N}$ be any $d$-tupling function with cubic shells.  If the non-negative integers $x_1$, $x_2$, \ldots, $x_d$ each have a binary representation with $n$ or fewer bits, then the binary representation of $f(x_1,x_2,\ldots,x_d)$ has $nd$ or fewer bits.
\end{theorem}
\begin{proof}
Consider any non-negative integers $x_1$, $x_2$, \ldots, $x_d$ whose binary representations each have $n$ or fewer bits.  Then, for each positive integer $i\leq d$,
\begin{align*}
\bigl\lceil\,\log_2(x_i+1)\bigr\rceil&\leq n,\\
\log_2(x_i+1)&\leq n,\\
x_i+1&\leq2^n.
\end{align*}
Letting $m=\max(x_1,x_2,\ldots,x_d)$, it immediately follows that
\begin{align*}
m+1&\leq2^n,\\
(m+1)^d&\leq2^{nd}.
\end{align*}
Then, by Corollary~\ref{cubic-inequality},
\begin{align*}
f(x_1,x_2,\ldots,x_d)<(m+1)^d&\leq2^{nd},\\
f(x_1,x_2,\ldots,x_d)+1&\leq2^{nd},\\
\log_2\bigl(f(x_1,x_2,\ldots,x_d)+1\bigr)&\leq nd,\\
\bigl\lceil\,\log_2\bigl(f(x_1,x_2,\ldots,x_d)+1\bigr)\bigr\rceil&\leq\bigl\lceil nd\,\bigr\rceil=nd.
\end{align*}
We may conclude that the binary representation of $f(x_1,x_2,\ldots,x_d)$ has $nd$ or fewer bits.\footnote{
If we replace the base $2$ with the base $b$ in this proof, where $b$ is any integer greater than $1$, then we obtain a proof of the following result: if the non-negative integers $x_1$, $x_2$, \ldots, $x_d$ each have a base-$b$ representation with $n$ or fewer digits, then the base-$b$ representation of $f(x_1,x_2,\ldots,x_d)$ has $nd$ or fewer digits.  This holds for all $d$-tupling functions $f\colon\mathbb{N}^d\to\mathbb{N}$ that have cubic shells.
}
\end{proof}

At the time of its publication, the results in Cantor's 1878 paper were surprising and unexpected~\citep{Dauben1979,Johnson1979}.
%
%
In addition to describing a $d$-tupling function for the unit interval of the real line, Cantor's paper was the first to provide an explicit formula for a pairing function for the positive integers.  Variants of Cantor's pairing function for the positive integers still dominate the literature, with variants of the Rosenberg-Strong pairing function tending to appear in more supplemental roles (for example, see Hrbacek and Jech~\citeyearpar{Hrbacek1978}, or Smory\'{n}ski~\citeyearpar{Smorynski1991}).  But Rosenberg and Strong were motivated by practical concerns, and their pairing function has practical advantages over Cantor's pairing function.  Indeed, we have seen that an enumeration of full binary trees that is constructed from the Rosenberg-Strong pairing function orders the trees in a more intuitively convenient sequence than the corresponding enumeration that is constructed from Cantor's pairing function.  And in implementations of pairing functions on contemporary computers, where careful attention is paid to the number of bits in the binary representation of each number, the Rosenberg-Strong pairing function also provides a practical advantage.  These advantages follow directly from the fact that the Rosenberg-Strong pairing function has cubic shells.  And in this way, we see some merit in the study of pairing functions and their shell numberings.

\section{Appendix}

In Section~\ref{additional}, it is claimed that the formula for $r_d^{-1}$, given in equations~\eqref{r-inv} and~\eqref{xd}, is the inverse of the Rosenberg-Strong $d$-tupling function.  This claim can be proved in the following manner.

\begin{lemma}\label{r-cubic}
For all positive integers $d$ and all $(x_1,x_2,\ldots,x_d)\in\mathbb{N}^d$,
\begin{equation*}
m^d\leq r_d(x_1,x_2,\ldots,x_d)<(m+1)^d,
\end{equation*}
where $m=\max(x_1,x_2,\ldots,x_d)$.
\end{lemma}
\begin{proof}
The proof is by induction on $d$.  For the base step, note that for all non-negative integers $x_1$, if $m=\max(x_1)$ then
\begin{equation*}
m^1\leq m=x_1=r_1(x_1)=x_1=m<(m+1)^1.
\end{equation*}
Now consider any integer $d>1$ and suppose, as the induction hypothesis, that for all $(x_1,\ldots,x_{d-1})\in\mathbb{N}^{d-1}$,
\begin{equation*}
n^{d-1}\leq r_{d-1}(x_1,\ldots,x_{d-1})<(n+1)^{d-1},
\end{equation*}
where $n=\max(x_1,\ldots,x_{d-1})$.  Next, consider any $(x_1,\ldots,x_{d-1},x_d)\in\mathbb{N}^d$ and let
%
%
$m=\max(x_1,\ldots,\linebreak[0]x_{d-1},x_d)$.  By the definition of $r_d$,
\begin{equation*}
r_d(x_1,\ldots,x_{d-1},x_d)=r_{d-1}(x_1,\ldots,x_{d-1})+m^d+(m-x_d)\bigl((m+1)^{d-1}-m^{d-1}\bigr).
\end{equation*}
But $r_{d-1}(x_1,\ldots,x_{d-1})$ and $(m-x_d)\bigl((m+1)^{d-1}-m^{d-1}\bigr)$ are non-negative integers, so
\begin{equation*}
r_d(x_1,\ldots,x_{d-1},x_d)\geq m^d.
\end{equation*}
And it follows from the definition of $r_d$, together with the induction hypothesis, that
\begin{align*}
r_d(x_1,\ldots,x_{d-1},x_d)&<(n+1)^{d-1}+m^d+(m-x_d)\bigl((m+1)^{d-1}-m^{d-1}\bigr),\\
r_d(x_1,\ldots,x_{d-1},x_d)&<(m+1)^{d-1}+m^d+(m-x_d)\bigl((m+1)^{d-1}-m^{d-1}\bigr),\\
r_d(x_1,\ldots,x_{d-1},x_d)&<(m+1)^{d-1}+m(m+1)^{d-1}-x_d\bigl((m+1)^{d-1}-m^{d-1}\bigr),\\
r_d(x_1,\ldots,x_{d-1},x_d)&<(m+1)^d-x_d\bigl((m+1)^{d-1}-m^{d-1}\bigr),\\
r_d(x_1,\ldots,x_{d-1},x_d)&<(m+1)^d.
\end{align*}
Hence, we have shown that
\begin{equation*}
m^d\leq r_d(x_1,\ldots,x_{d-1},x_d)<(m+1)^d.
\end{equation*}
\end{proof}

\begin{corollary}\label{r-root}
Let $d$ be a positive integer.  For all $(x_1,x_2,\ldots,x_d)\in\mathbb{N}^d$,
\begin{equation*}
\max(x_1,x_2,\ldots,x_d)=\bigl\lfloor\sqrt[d]{r_d(x_1,x_2,\ldots,x_d)}\,\bigr\rfloor.
\end{equation*}
\end{corollary}
\begin{remark*}
A simple generalization of the following proof shows that a $d$-tupling function $f\colon\mathbb{N}^d\to\mathbb{N}$ has cubic shells if and only if $\max(\mathbf{x})=\bigl\lfloor\sqrt[d]{f(\mathbf{x})}\,\bigr\rfloor$ for all $\mathbf{x}\in\mathbb{N}^d$.
\end{remark*}
\begin{proof}
Consider any $(x_1,x_2,\ldots,x_d)\in\mathbb{N}^d$ and let $m=\max(x_1,x_2,\ldots,x_d)$.  By Lemma~\ref{r-cubic},
\begin{equation*}
m^d\leq r_d(x_1,x_2,\ldots,x_d)<(m+1)^d.
\end{equation*}
Therefore,
\begin{equation*}
m\leq\sqrt[d]{r_d(x_1,x_2,\ldots,x_d)}<m+1,
\end{equation*}
and we may conclude that $m=\bigl\lfloor\sqrt[d]{r_d(x_1,x_2,\ldots,x_d)}\,\bigr\rfloor$.
\end{proof}

\begin{lemma}\label{xd-inequality}
Given any positive integer $d$ and any non-negative integer $z$, let $m=\bigl\lfloor\sqrt[d]{z}\,\bigr\rfloor$ and let $x_d$ be defined by equation~\eqref{xd}.  Then,
\begin{equation*}
0\leq x_d\leq m.
\end{equation*}
\end{lemma}
\begin{proof}
Given any positive integer $d$ and any non-negative integer $z$, let $m=\bigl\lfloor\sqrt[d]{z}\,\bigr\rfloor$.  Then,
\begin{align*}
\sqrt[d]{z}&<m+1,\\
z&<(m+1)^d,\\
z-(m+1)m^{d-1}&<(m+1)^d-(m+1)m^{d-1},\\
z-m^d-m^{d-1}&<\bigl(m+1\bigr)\bigl((m+1)^{d-1}-m^{d-1}\bigr),\\
0\leq\max\bigl(0,z-m^d-m^{d-1}\bigr)&<\bigl(m+1\bigr)\bigl((m+1)^{d-1}-m^{d-1}\bigr).
\end{align*}
And dividing by $(m+1)^{d-1}-m^{d-1}$,
\begin{alignat*}{3}
0&\leq{}&&\frac{\max\bigl(0,z-m^d-m^{d-1}\bigr)}{(m+1)^{d-1}-m^{d-1\rule{0pt}{6pt}}}&&<m+1,\\
0&\leq{}&\Biggl\lfloor&\frac{\max\bigl(0,z-m^d-m^{d-1}\bigr)}{(m+1)^{d-1}-m^{d-1\rule{0pt}{6pt}}}\Biggr\rfloor&&<m+1,\\
0&\leq{}&\Biggl\lfloor&\frac{\max\bigl(0,z-m^d-m^{d-1}\bigr)}{(m+1)^{d-1}-m^{d-1\rule{0pt}{6pt}}}\Biggr\rfloor&&\leq m,\\
0&\geq{}&-\Biggl\lfloor&\frac{\max\bigl(0,z-m^d-m^{d-1}\bigr)}{(m+1)^{d-1}-m^{d-1\rule{0pt}{6pt}}}\Biggr\rfloor&&\geq-m,\\
m&\geq{}&m-\Biggl\lfloor&\frac{\max\bigl(0,z-m^d-m^{d-1}\bigr)}{(m+1)^{d-1}-m^{d-1\rule{0pt}{6pt}}}\Biggr\rfloor&&\geq0.
\end{alignat*}
We may conclude from equation~\eqref{xd} that
\begin{equation*}
m\geq x_d\geq0.
\end{equation*}
\end{proof}

\begin{lemma}\label{r-right-inverse}
Let $d$ be any positive integer.  The formula for $r_d^{-1}(z)$ that is given in equations~\eqref{r-inv} and~\eqref{xd} describes a function from $\mathbb{N}$ to $\mathbb{N}^d$.  Moreover,
\begin{equation*}
r_d\bigl(r_d^{-1}(z)\bigr)=z
\end{equation*}
for all $z\in\mathbb{N}$.
\end{lemma}
\begin{proof}
The proof is by induction on $d$.  For the base step, note that $r_1^{-1}(z)=z$ is a function from $\mathbb{N}$ to $\mathbb{N}$, and
\begin{equation*}
r_1\bigl(r_1^{-1}(z)\bigr)=r_1(z)=z
\end{equation*}
for all $z\in\mathbb{N}$.  Next, consider any integer $d>1$ and suppose, as the induction hypothesis, that $r_{d-1}^{-1}$ is a function from $\mathbb{N}$ to $\mathbb{N}^{d-1}$ and
\begin{equation*}
r_{d-1}\bigl(r_{d-1}^{-1}(z)\bigr)=z
\end{equation*}
for all $z\in\mathbb{N}$.  Now consider any non-negative integer $z$ and let $m=\bigl\lfloor\sqrt[d]{z}\,\bigr\rfloor$.  By equation~\eqref{r-inv},
\begin{equation}\label{r-inv-u}
r_d^{-1}(z)=\bigl(r_{d-1}^{-1}(u),x_d\bigr),
\end{equation}
where $x_d$ is defined by equation~\eqref{xd}, and where
\begin{equation*}
u=z-m^d-(m-x_d)\bigl((m+1)^{d-1}-m^{d-1}\bigr).
\end{equation*}
Note that if $u\in\mathbb{N}$, then it follows from the induction hypothesis that $r_{d-1}^{-1}(u)$ is a member of $\mathbb{N}^{d-1}$.  Consequently, by equation~\eqref{r-inv-u} and Lemma~\ref{xd-inequality}, $r_d^{-1}(z)$ is a member of $\mathbb{N}^d$ if $u\in\mathbb{N}$.  And by Corollary~\ref{r-root},
%
%
\begin{equation*}
u\in\mathbb{N}\quad\text{implies}\quad\max\bigl(r_{d-1}^{-1}(u)\bigr)=\biggl\lfloor\Bigl(r_{d-1}\bigl(r_{d-1}^{-1}(u)\bigr)\Bigr)^{1/(d-1)}\biggr\rfloor.
\end{equation*}
Hence, by the induction hypothesis,
\begin{equation}\label{r-inv-root}
u\in\mathbb{N}\quad\text{implies}\quad\max\bigl(r_{d-1}^{-1}(u)\bigr)=\Bigl\lfloor u^{1/(d-1)}\Bigr\rfloor.
\end{equation}
There are now two cases to consider.
\begin{description}
\item[Case 1] If $z-m^d-m^{d-1}<0$, then $x_d=m$ by equation~\eqref{xd}.  Therefore, $u=z-m^d$.  But $m=\bigl\lfloor\sqrt[d]{z}\,\bigr\rfloor$, so
\begin{align*}
m&\leq\sqrt[d]{z},\\
m^d&\leq z.
\end{align*}
It immediately follows that $u=z-m^d$ is a non-negative integer.  And by condition~\eqref{r-inv-root},
\begin{equation*}
\max\bigl(r_{d-1}^{-1}(u)\bigr)=\Bigl\lfloor\bigl(z-m^d\,\bigr)^{1/(d-1)}\Bigr\rfloor.
\end{equation*}
But because $z-m^d-m^{d-1}<0$,
\begin{align*}
z-m^d&<m^{d-1},\\
\bigl(z-m^d\,\bigr)^{1/(d-1)}&<m,\\
\Bigl\lfloor\bigl(z-m^d\,\bigr)^{1/(d-1)}\Bigr\rfloor&<m.
\end{align*}
Hence, $\max\bigl(r_{d-1}^{-1}(u)\bigr)<m$.  And since $x_d=m$, it follows from equation~\eqref{r-inv-u} that
%
%
$\max\bigl(r_d^{-1}(z)\bigr)\linebreak[0]=m$.
\item[Case 2] If $z-m^d-m^{d-1}\geq 0$ then by equation~\eqref{xd},
\begin{equation*}
x_d=m-\Biggl\lfloor\frac{z-m^d-m^{d-1}}{(m+1)^{d-1}-m^{d-1\rule{0pt}{6pt}}}\Biggr\rfloor.
\end{equation*}
But by the Euclidean division theorem,
\begin{equation*}
z-m^d-m^{d-1}=\Biggl\lfloor\frac{z-m^d-m^{d-1}}{(m+1)^{d-1}-m^{d-1\rule{0pt}{6pt}}}\Biggr\rfloor\bigl((m+1)^{d-1}-m^{d-1}\bigr)+b
\end{equation*}
for some integer $b$ such that $0\leq b<(m+1)^{d-1}-m^{d-1}$.  Thus,
\begin{equation*}
z-m^d-m^{d-1}=(m-x_d)\bigl((m+1)^{d-1}-m^{d-1}\bigr)+b.
\end{equation*}
It immediately follows that
\begin{equation*}
z-m^d-(m-x_d)\bigl((m+1)^{d-1}-m^{d-1}\bigr)=b+m^{d-1}.
\end{equation*}
Therefore, $u=b+m^{d-1}$.  And by condition~\eqref{r-inv-root},
\begin{equation*}
\max\bigl(r_{d-1}^{-1}(u)\bigr)=\Bigl\lfloor\bigl(b+m^{d-1}\bigr)^{1/(d-1)}\Bigr\rfloor.
\end{equation*}
Moreover, since
\begin{equation*}
0\leq b<(m+1)^{d-1}-m^{d-1},
\end{equation*}
we have that
\begin{alignat*}{5}
&m^{d-1}&&\leq{}&b&+m^{d-1}&&<{}&(&m+1)^{d-1},\\
&m&&\leq{}&\bigl(b+{}&m^{d-1}\bigr)^{1/(d-1)}&&<{}&&m+1.
\end{alignat*}
Hence,
\begin{equation*}
\max\bigl(r_{d-1}^{-1}(u)\bigr)=\Bigl\lfloor\bigl(b+m^{d-1}\bigr)^{1/(d-1)}\Bigr\rfloor=m.
\end{equation*}
It then follows from Lemma~\ref{xd-inequality} and equation~\eqref{r-inv-u} that $\max\bigl(r_d^{-1}(z)\bigr)=m$.
\end{description}
In either case, $u\in\mathbb{N}$ and $\max\bigl(r_d^{-1}(z)\bigr)=m$.  So, by equation~\eqref{r-inv-u} and the definition of $r_d$,
\begin{align*}
r_d\bigl(r_d^{-1}(z)\bigr)&=r_d\bigl(r_{d-1}^{-1}(u),x_d\bigr)\\
&=r_{d-1}\bigl(r_{d-1}^{-1}(u)\bigr)+m^d+(m-x_d)\bigl((m+1)^{d-1}-m^{d-1}\bigr).
\end{align*}
Then, by the induction hypothesis,
\begin{equation*}
r_d\bigl(r_d^{-1}(z)\bigr)=u+m^d+(m-x_d)\bigl((m+1)^{d-1}-m^{d-1}\bigr).
\end{equation*}
And by the definition of $u$,
\begin{align*}
r_d\bigl(r_d^{-1}(z)\bigr)={}&z-m^d-(m-x_d)\bigl((m+1)^{d-1}-m^{d-1}\bigr)\\
&+m^d+(m-x_d)\bigl((m+1)^{d-1}-m^{d-1}\bigr).
\end{align*}
It immediately follows that
\begin{equation*}
r_d\bigl(r_d^{-1}(z)\bigr)=z.
\end{equation*}
Moreover, $u\in\mathbb{N}$ for all non-negative integers $z$.  Therefore, $r_d^{-1}(z)\in\mathbb{N}^d$ for all non-negative integers $z$.  We may conclude that $r_d^{-1}$ is a function from $\mathbb{N}$ to $\mathbb{N}^d$.
\end{proof}

\begin{lemma}\label{r-left-inverse}
Let $d$ be a positive integer.  For all $(x_1,x_2,\ldots,x_d)\in\mathbb{N}^d$,
\begin{equation*}
r_d^{-1}\bigl(r_d(x_1,x_2,\ldots,x_d)\bigr)=(x_1,x_2,\ldots,x_d).
\end{equation*}
\end{lemma}
\begin{proof}
The proof is by induction on $d$.  For the base step, note that
\begin{equation*}
r_1^{-1}\bigl(r_1(x_1)\bigr)=r_1^{-1}(x_1)=x_1
\end{equation*}
for all $x_1\in\mathbb{N}$.  Next, consider any integer $d>1$ and suppose, as the induction hypothesis, that
\begin{equation*}
r_{d-1}^{-1}\bigl(r_{d-1}(x_1,\ldots,x_{d-1})\bigr)=(x_1,\ldots,x_{d-1})
\end{equation*}
for all $(x_1,\ldots,x_{d-1})\in\mathbb{N}^{d-1}$.  Now consider any $(x_1,\ldots,x_{d-1},x_d)\in\mathbb{N}^d$.  Let
%
%
$z=r_d(x_1,\ldots,x_{d-1},\linebreak[0]x_d)$ and let $m=\max(x_1,\ldots,x_{d-1},x_d)$.  By the definition of $r_d$,
\begin{equation*}
z=r_{d-1}(x_1,\ldots,x_{d-1})+m^d+(m-x_d)\bigl((m+1)^{d-1}-m^{d-1}\bigr).
\end{equation*}
Hence,
\begin{equation}\label{r-z}
z-m^d-(m-x_d)\bigl((m+1)^{d-1}-m^{d-1}\bigr)=r_{d-1}(x_1,\ldots,x_{d-1}).
\end{equation}
There are two cases to consider.
\begin{description}
\item[Case 1] If $m=x_d>\max(x_1,\ldots,x_{d-1})$ then it follows from equation~\eqref{r-z} that
\begin{equation*}
z-m^d=r_{d-1}(x_1,\ldots,x_{d-1}).
\end{equation*}
Now let $n=\max(x_1,\ldots,x_{d-1})$.  By Lemma~\ref{r-cubic},
\begin{equation*}
z-m^d<(n+1)^{d-1}.
\end{equation*}
But $m>n=\max(x_1,\ldots,x_{d-1})$, so $m\geq n+1$.  Therefore,
\begin{align*}
z-m^d&<m^{d-1},\\
z-m^d-m^{d-1}&<0.
\end{align*}
It immediately follows that
\begin{equation*}
x_d=m=m-\Biggl\lfloor\frac{\max\bigl(0,z-m^d-m^{d-1}\bigr)}{(m+1)^{d-1}-m^{d-1\rule{0pt}{6pt}}}\Biggr\rfloor.
\end{equation*}
\item[Case 2] If $m=\max(x_1,\ldots,x_{d-1})$ then it follows from Lemma~\ref{r-cubic} and equation~\eqref{r-z} that
\begin{equation*}
m^{d-1}\leq z-m^d-(m-x_d)\bigl((m+1)^{d-1}-m^{d-1}\bigr).
\end{equation*}
Therefore,
\begin{equation*}
-z+m^d+m^{d-1}\leq-(m-x_d)\bigl((m+1)^{d-1}-m^{d-1}\bigr)
\end{equation*}
and
\begin{equation}\label{xd-low}
\frac{z-m^d-m^{d-1}}{(m+1)^{d-1}-m^{d-1\rule{0pt}{6pt}}}\geq m-x_d.
\end{equation}
It also follows from Lemma~\ref{r-cubic} and equation~\eqref{r-z} that
\begin{equation*}
z-m^d-(m-x_d)\bigl((m+1)^{d-1}-m^{d-1}\bigr)<(m+1)^{d-1}.
\end{equation*}
Hence,
\begin{equation*}
-(m-x_d)\bigl((m+1)^{d-1}-m^{d-1}\bigr)<-z+m^d+(m+1)^{d-1}
\end{equation*}
and
\begin{align}
m-x_d&>\frac{z-m^d-(m+1)^{d-1}}{(m+1)^{d-1}-m^{d-1\rule{0pt}{6pt}}},\notag\\
m-x_d&>\frac{z-m^d-m^{d-1}-\bigl((m+1)^{d-1}-m^{d-1}\bigr)}{(m+1)^{d-1}-m^{d-1\rule{0pt}{6pt}}},\notag\\
m-x_d&>\frac{z-m^d-m^{d-1}}{(m+1)^{d-1}-m^{d-1\rule{0pt}{6pt}}}-1.\label{xd-high}
\end{align}
Combining inequalities~\eqref{xd-low} and~\eqref{xd-high},
\begin{equation*}
\frac{z-m^d-m^{d-1}}{(m+1)^{d-1}-m^{d-1\rule{0pt}{6pt}}}\geq m-x_d>\frac{z-m^d-m^{d-1}}{(m+1)^{d-1}-m^{d-1\rule{0pt}{6pt}}}-1.
\end{equation*}
It immediately follows that
\begin{equation*}
m-x_d=\Biggl\lfloor\frac{z-m^d-m^{d-1}}{(m+1)^{d-1}-m^{d-1\rule{0pt}{6pt}}}\Biggr\rfloor.
\end{equation*}
But $m-x_d$ cannot be negative because $m=\max(x_1,\ldots,x_{d-1},x_d)$.  Therefore,
\begin{equation*}
m-x_d=\Biggl\lfloor\frac{\max\bigl(0,z-m^d-m^{d-1}\bigr)}{(m+1)^{d-1}-m^{d-1\rule{0pt}{6pt}}}\Biggr\rfloor.
\end{equation*}
\end{description}
In either case,
\begin{equation*}
x_d=m-\Biggl\lfloor\frac{\max\bigl(0,z-m^d-m^{d-1}\bigr)}{(m+1)^{d-1}-m^{d-1\rule{0pt}{6pt}}}\Biggr\rfloor.
\end{equation*}
But $m=\bigl\lfloor\sqrt[d]{z}\,\bigr\rfloor$ by Corollary~\ref{r-root}.  Hence, by equation~\eqref{r-inv}, equation~\eqref{r-z}, and the induction hypothesis,
\begin{multline*}
r_d^{-1}\bigl(r_d(x_1,\ldots,x_{d-1},x_d)\bigr)=r_d^{-1}(z)\\
\begin{aligned}
&=\Bigl(r_{d-1}^{-1}\bigl(z-m^d-(m-x_d)((m+1)^{d-1}-m^{d-1})\bigr),x_d\Bigr)\\
&=\Bigl(r_{d-1}^{-1}\bigl(r_{d-1}(x_1,\ldots,x_{d-1})\bigr),x_d\Bigr)\\
&=(x_1,\ldots,x_{d-1},x_d).
\end{aligned}
\end{multline*}
\end{proof}

\begin{theorem}
For each positive integer $d$, $r_d\colon\mathbb{N}^d\to\mathbb{N}$ is a $d$-tupling function that has $r_d^{-1}\colon\mathbb{N}\to\mathbb{N}^d$ as its inverse.
\end{theorem}
\begin{proof}
By Lemma~\ref{r-right-inverse}, $r_d^{-1}$ is a right inverse of $r_d$.  And by Lemma~\ref{r-left-inverse}, it is a left inverse of $r_d$.  Therefore, $r_d^{-1}$ is the inverse of $r_d$.  And since invertible functions are one-to-one correspondences, $r_d$ is a $d$-tupling function for the non-negative integers.
\end{proof}

\section{Acknowledgements}

We thank John P. D'Angelo, Douglas B. West, and Alexander Yong for the information that they provided regarding Macaulay representations.  We thank the user \texttt{nawfal} at
\begin{center}
\texttt{http://stackoverflow.com/questions/919612\#answer-13871379}
\end{center}
for introducing us to a special case of Theorem~\ref{cubic-bits}.  We also thank Dana S. Scott, Fritz H. Obermeyer, and Machi Zawidzki for their advice.

\bibliographystyle{abbrvnat}
\bibliography{RosenbergStrong}

\begin{thebibliography}{47}
\providecommand{\natexlab}[1]{#1}
\providecommand{\url}[1]{\texttt{#1}}
\expandafter\ifx\csname urlstyle\endcsname\relax
  \providecommand{\doi}[1]{doi: #1}\else
  \providecommand{\doi}{doi: \begingroup \urlstyle{rm}\Url}\fi

\bibitem[Barendregt(1974)]{Barendregt1974}
H.~Barendregt.
\newblock Pairing without conventional restraints.
\newblock \emph{Zeitschrift f{\"{u}}r mathematische Logik und Grundlagen der
  Mathematik}, 20:\penalty0 289--306, 1974.
\newblock \doi{10.1002/malq.19740201902}.

\bibitem[Bradley and Sandifer(2009)]{Bradley2009}
R.~E. Bradley and C.~E. Sandifer.
\newblock \emph{Cauchy's Cours d'analyse}.
\newblock Springer, New York, 2009.
\newblock \doi{10.1007/978-1-4419-0549-9}.

\bibitem[Cantor(1878)]{Cantor1878}
G.~Cantor.
\newblock Ein {B}eitrag zur {M}annigfaltigkeitslehre.
\newblock \emph{Journal f{\"{u}}r die reine und angewandte Mathematik},
  84:\penalty0 242--258, 1878.

\bibitem[Cauchy(1821)]{Cauchy1821}
A.-L. Cauchy.
\newblock \emph{Cours d'analyse de l'{{\'{E}}}cole royale polytechnique}.
\newblock Debure fr{\`{e}}res, Paris, 1821.
\newblock See Bradley and Sandifer~\citeyearpar{Bradley2009} for an English
  translation.

\bibitem[Chowla(1961)]{Chowla1961}
P.~Chowla.
\newblock On some polynomials which represent every natural number exactly
  once.
\newblock \emph{Det Kongelige Norske Videnskabers Selskabs Forhandlinger},
  34\penalty0 (2):\penalty0 8--9, 1961.
\newblock \emph{Erratum}: replace $x_j$ with $x_{j+1}$ in the last formula on
  page 9.

\bibitem[Cutland(1980)]{Cutland1980}
N.~Cutland.
\newblock \emph{Computability: An Introduction to Recursive Function Theory}.
\newblock Cambridge University Press, Cambridge, 1980.

\bibitem[Dau{\-}ben(1979)]{Dauben1979}
J.~W. Dau{\-}ben.
\newblock \emph{Georg {C}antor: His Mathematics and Philosophy of the
  Infinite}.
\newblock Harvard University Press, Cambridge, Massachusetts, 1979.

\bibitem[Davis(1958)]{Davis1958}
M.~D. Davis.
\newblock \emph{Computability and Unsolvability}.
\newblock McGraw-Hill, New York, 1958.

\bibitem[Davis and Weyuker(1983)]{Davis1983}
M.~D. Davis and E.~J. Weyuker.
\newblock \emph{Computability, Complexity, and Languages}.
\newblock Academic Press, Orlando, Florida, 1983.

\bibitem[Enderton(1977)]{Enderton1977}
H.~B. Enderton.
\newblock \emph{Elements of Set Theory}.
\newblock Academic Press, San Diego, California, 1977.

\bibitem[Frankl(1984)]{Frankl1984}
P.~Frankl.
\newblock A new short proof for the {K}ruskal-{K}atona theorem.
\newblock \emph{Discrete Mathematics}, 48\penalty0 (2--3):\penalty0 327--329,
  1984.
\newblock \doi{10.1016/0012-365X(84)90193-6}.

\bibitem[Fueter and P{\'{o}}lya(1923)]{Fueter1923}
R.~Fueter and G.~P{\'{o}}lya.
\newblock Rationale {A}bz{\"{a}}hlung der {G}itterpunkte.
\newblock \emph{Vierteljahrsschrift der Naturforschenden Gesellschaft in
  Z{\"{u}}rich}, 68\penalty0 (3--4):\penalty0 380--386, Dec. 1923.

\bibitem[Green(1989)]{Green1989}
M.~Green.
\newblock Restrictions of linear series to hyperplanes, and some results of
  {M}acaulay and {G}otzmann.
\newblock In \emph{Algebraic Curves and Projective Geometry}, pages 76--86,
  Berlin, 1989. Springer-Verlag.
\newblock \doi{10.1007/BFb0085925}.

\bibitem[Greene and Kleitman(1978)]{Greene1978}
C.~Greene and D.~J. Kleitman.
\newblock Proof techniques in the theory of finite sets.
\newblock In G.-C. Rota, editor, \emph{Studies in Combinatorics}, pages 22--79.
  Mathematical Association of America, Washington, D.C., 1978.

\bibitem[Hausdorff(1927)]{Hausdorff1927}
F.~Hausdorff.
\newblock \emph{Mengenlehre}.
\newblock Walter de Gruyter, Berlin, second edition, 1927.
\newblock See Hausdorff~\citeyearpar{Hausdorff1957} for an English translation.

\bibitem[Hausdorff(1957)]{Hausdorff1957}
F.~Hausdorff.
\newblock \emph{Set Theory}.
\newblock Chelsea, New York, 1957.

\bibitem[Hrbacek and Jech(1978)]{Hrbacek1978}
K.~Hrbacek and T.~Jech.
\newblock \emph{Introduction to Set Theory}.
\newblock Marcel Dekker, New York, 1978.

\bibitem[Johnson(1979)]{Johnson1979}
D.~M. Johnson.
\newblock The problem of the invariance of dimension in the growth of modern
  topology, {P}art~{I}.
\newblock \emph{Archive for History of Exact Sciences}, 20\penalty0
  (2):\penalty0 97--188, 1979.
\newblock \doi{10.1007/BF00327627}.

\bibitem[Knuth(1997)]{Knuth1997}
D.~E. Knuth.
\newblock \emph{The Art of Computer Programming}, volume~1.
\newblock Addison-Wesley, Upper Saddle River, New Jersey, third edition, 1997.

\bibitem[Kruskal(1963)]{Kruskal1963}
J.~B. Kruskal.
\newblock The number of simplices in a complex.
\newblock In R.~Bellman, editor, \emph{Mathematical Optimization Techniques},
  pages 251--278. University of California Press, Berkeley, California, 1963.

\bibitem[Lew(1979)]{Lew1979}
J.~S. Lew.
\newblock Polynomial enumeration of multidimensional lattices.
\newblock \emph{Mathematical Systems Theory}, 12:\penalty0 253--270, 1979.
\newblock \doi{10.1007/BF01776577}.

\bibitem[Lew and Rosenberg(1978{\natexlab{a}})]{Lew1978a}
J.~S. Lew and A.~L. Rosenberg.
\newblock Polynomial indexing of integer lattice-points {I}. {G}eneral concepts
  and quadratic polynomials.
\newblock \emph{Journal of Number Theory}, 10\penalty0 (2):\penalty0 192--214,
  1978{\natexlab{a}}.
\newblock \doi{10.1016/0022-314X(78)90035-5}.

\bibitem[Lew and Rosenberg(1978{\natexlab{b}})]{Lew1978b}
J.~S. Lew and A.~L. Rosenberg.
\newblock Polynomial indexing of integer lattice-points {I}{I}. {N}onexistence
  results for higher-degree polynomials.
\newblock \emph{Journal of Number Theory}, 10\penalty0 (2):\penalty0 215--243,
  1978{\natexlab{b}}.
\newblock \doi{10.1016/0022-314X(78)90036-7}.

\bibitem[Matiyasevich(1993)]{Matiyasevich1993}
Y.~V. Matiyasevich.
\newblock \emph{Hilbert's Tenth Problem}.
\newblock MIT Press, Cambridge, Massachussetts, 1993.

\bibitem[Minsky(1967)]{Minsky1967}
M.~L. Minsky.
\newblock \emph{Computation: Finite and Infinite Machines}.
\newblock Prentice-Hall, Englewood Cliffs, New Jersey, 1967.

\bibitem[Morales and Arredondo~R.(1999)]{Morales1999}
L.~B. Morales and J.~H. Arredondo~R.
\newblock A family of asymptotically $e(n-1)!$ polynomial orders of {$N^n$}.
\newblock \emph{Order}, 16\penalty0 (2):\penalty0 195--206, 1999.
\newblock \doi{10.1023/A:1006329224202}.
\newblock \emph{Erratum}: the table on page 205 contains multiple errors.

\bibitem[Morales and Lew(1996)]{Morales1996}
L.~B. Morales and J.~S. Lew.
\newblock An enlarged family of packing polynomials on multidimensional
  lattices.
\newblock \emph{Mathematical Systems Theory}, 29\penalty0 (3):\penalty0
  293--303, 1996.
\newblock \doi{10.1007/BF01201281}.

\bibitem[Pascal(1887)]{Pascal1887}
E.~Pascal.
\newblock Sopra una formola numerica.
\newblock \emph{Giornale di Matematiche}, 25:\penalty0 45--49, 1887.

\bibitem[P{\'{e}}ter(1951)]{Peter1951}
R.~P{\'{e}}ter.
\newblock \emph{Rekursive Funktionen}.
\newblock Akad\'{e}miai Kiad\'{o}, Budapest, 1951.
\newblock See P{\'{e}}ter~\citeyearpar{Peter1967} for an English translation.

\bibitem[P{\'{e}}ter(1967)]{Peter1967}
R.~P{\'{e}}ter.
\newblock \emph{Recursive Functions}.
\newblock Academic Press, New York, 1967.

\bibitem[Regan(1992)]{Regan1992}
K.~W. Regan.
\newblock Minimum-complexity pairing functions.
\newblock \emph{Journal of Computer and System Sciences}, 45\penalty0
  (3):\penalty0 285--295, Dec. 1992.
\newblock \doi{10.1016/0022-0000(92)90027-G}.
\newblock \emph{Erratum}: equation~(4.2) does not define a pairing function for
  the positive integers.

\bibitem[Rogers(1967)]{Rogers1967}
H.~Rogers, Jr.
\newblock \emph{Theory of Recursive Functions and Effective Computability}.
\newblock McGraw-Hill, New York, 1967.

\bibitem[Rosenberg(1974)]{Rosenberg1974}
A.~L. Rosenberg.
\newblock Allocating storage for extendible arrays.
\newblock \emph{Journal of the ACM}, 21\penalty0 (4):\penalty0 652--670, Oct.
  1974.
\newblock \doi{10.1145/321850.321861}.
\newblock Corrigendum in Rosenberg~\citeyearpar{Rosenberg1975a}.

\bibitem[Rosenberg(1975{\natexlab{a}})]{Rosenberg1975a}
A.~L. Rosenberg.
\newblock Corrigendum.
\newblock \emph{Journal of the ACM}, 22\penalty0 (2):\penalty0 308, Apr.
  1975{\natexlab{a}}.
\newblock \doi{10.1145/321879.321891}.

\bibitem[Rosenberg(1975{\natexlab{b}})]{Rosenberg1975b}
A.~L. Rosenberg.
\newblock Managing storage for extendible arrays.
\newblock \emph{SIAM Journal on Computing}, 4\penalty0 (3):\penalty0 287--306,
  1975{\natexlab{b}}.
\newblock \doi{10.1137/0204024}.

\bibitem[Rosenberg(1978)]{Rosenberg1978}
A.~L. Rosenberg.
\newblock Storage mappings for extendible arrays.
\newblock In R.~T. Yeh, editor, \emph{Current Trends in Programming
  Methodology: Data Structuring}, volume~4, pages 263--311. Prentice-Hall,
  Englewood Cliffs, New Jersey, 1978.

\bibitem[Rosenberg(2003)]{Rosenberg2003}
A.~L. Rosenberg.
\newblock Efficient pairing functions --- and why you should care.
\newblock \emph{International Journal of Foundations of Computer Science},
  14\penalty0 (1):\penalty0 3--17, 2003.
\newblock \doi{10.1142/S012905410300156X}.

\bibitem[Rosenberg and Strong(1972)]{Rosenberg1972}
A.~L. Rosenberg and H.~R. Strong.
\newblock Addressing arrays by shells.
\newblock \emph{IBM Technical Disclosure Bulletin}, 14\penalty0 (10):\penalty0
  3026--3028, Mar. 1972.

\bibitem[Rosser(1937)]{Rosser1937}
B.~Rosser.
\newblock Th. {S}kolem. {{\"{U}}}ber die {Z}ur{\"{u}}ckf{\"{u}}hrbarkeit
  einiger durch {R}ekursionen definierter {R}elationen auf
  {``}arithmetische{''}.
\newblock \emph{The Journal of Symbolic Logic}, 2\penalty0 (2):\penalty0
  85--86, June 1937.
\newblock \doi{10.2307/2267375}.

\bibitem[Sierpi{\'{n}}ski(1912)]{Sierpinski1912}
W.~Sierpi{\'{n}}ski.
\newblock \emph{Zarys Teoryi Mnogo{\'{s}}ci}.
\newblock Ksi{\k{e}}garnia E. Wendego, War\-szawa, 1912.

\bibitem[Simpson(1999)]{Simpson1999}
S.~G. Simpson.
\newblock \emph{Subsystems of Second Order Arithmetic}.
\newblock Springer-Verlag, Berlin, 1999.

\bibitem[Skolem(1937)]{Skolem1937}
T.~Skolem.
\newblock {\"{U}}ber die {Z}ur{\"{u}}ckf{\"{u}}hrbarkeit einiger durch
  {R}ekursionen defi\-nierter {R}elationen auf {``}arithmetische{''}.
\newblock \emph{Acta Scientiarum Mathematicarum (Szeged)}, 8\penalty0
  (2--3):\penalty0 73--88, 1937.
\newblock See Rosser~\citeyearpar{Rosser1937} for an English language summary.

\bibitem[Smory{\'{n}}ski(1991)]{Smorynski1991}
C.~Smory{\'{n}}ski.
\newblock \emph{Logical Number Theory {I}}.
\newblock Springer-Verlag, Berlin, 1991.
\newblock \doi{10.1007/978-3-642-75462-3}.

\bibitem[Tajtelbaum-Tarski(1924)]{Tarski1924}
A.~Tajtelbaum-Tarski.
\newblock Sur quelques th{\'{e}}or{\`{e}}mes qui {\'{e}}quivalent {\`{a}}
  l'axiome du choix.
\newblock \emph{Fundamenta Mathematicae}, 5:\penalty0 147--154, 1924.
\newblock \doi{10.4064/fm-5-1-147-154}.

\bibitem[West(1996)]{West1996}
D.~B. West.
\newblock \emph{Introduction to Graph Theory}.
\newblock Prentice-Hall, Upper Saddle River, New Jersey, 1996.

\bibitem[Zermelo(1904)]{Zermelo1904}
E.~Zermelo.
\newblock Beweis, da{\ss} jede {M}enge wohlgeordnet werden kann.
\newblock \emph{Mathematische Annalen}, 59\penalty0 (4):\penalty0 514--516,
  Nov. 1904.
\newblock \doi{10.1007/BF01445300}.
\newblock See Zermelo~\citeyearpar{Zermelo1967} for an English translation.

\bibitem[Zermelo(1967)]{Zermelo1967}
E.~Zermelo.
\newblock Proof that every set can be well-ordered.
\newblock In J.~van Heijenoort, editor, \emph{From Frege to G{\"{o}}del}, pages
  139--141. Harvard University Press, Cambridge, Massachussetts, 1967.

\end{thebibliography}

\end{document}